\documentclass[smallextended]{svjour3}

\smartqed 

\usepackage{graphicx,amssymb,amsmath,color}
\usepackage{hyperref}

\usepackage{dcolumn}
\usepackage{bm}
\usepackage{mathtools} 
\usepackage{mathrsfs}
\usepackage{yfonts}
\usepackage[makeroom]{cancel}

\usepackage{txfonts}

\DeclareMathOperator*{\opls}{\scalerel*{\oplus}{\textstyle\sum}}
\DeclareMathOperator*{\Span}{\scalerel*{\mathtt{Span}}{\textstyle\sum}}

\DeclareMathOperator*{\Union}{\scalerel*{\cup}{\textstyle\sum}}

\usepackage{scalerel}

\usepackage{bbm}

\newcommand{\hs}[0]{\hspace*{3.5mm}}
\newcommand{\be}[0]{\begin{eqnarray*}}
	\newcommand{\ee}[0]{\end{eqnarray*}}
\newcommand{\ben}[0]{\begin{eqnarray}}
\newcommand{\een}[0]{\end{eqnarray}}

\newcommand{\simx}[1]{\substack{\phantom{a}_{#1}\phantom{a}\\\widetilde{\phantom{#1}}}}
\newcommand{\Q}{\ensuremath{\boldsymbol{\mathtt{Q}}}}
\newcommand{\RR}{\ensuremath{\boldsymbol{\mathtt{R}}}}
\newcommand{\SSS}{\ensuremath{\boldsymbol{\mathtt{S}}}}
\newcommand{\aff}{\ensuremath{\mathtt{\alpha}}}	
\newcommand{\F}{\ensuremath{\boldsymbol{\mathtt{F}}}}	
\newcommand{\LL}{\ensuremath{\boldsymbol{\mathtt{L}}}}	
\newcommand{\HH}{\ensuremath{\boldsymbol{\mathtt{H}}}}	
\newcommand{\bgm}{\ensuremath{\boldsymbol{\mathtt{\gamma}}}}	
\newcommand{\LTr}{\ensuremath{\mathbb{F}\LL}}	
\newcommand{\Hmap}{\ensuremath{\mathbb{F}\HH}}	
\newcommand{\Saction}{\ensuremath{\mathtt{S}}} 	
\newcommand{\Ii}{\ensuremath{\mathtt{I}}} 	
\newcommand{\JJ}{\ensuremath{\mathcal{J}}} 	

\newcommand{\vb}[1]{\ensuremath{|_{_{#1}}}}		
\newcommand{\binfty}{\ensuremath{\boldsymbol{\mathtt{\infty}}}}
\newcommand{\id}[0]{\ensuremath{\text{id}}}

\newcommand{\bt}{\ensuremath{\boldsymbol{t}}} 
\newcommand{\R}[0]{\ensuremath{\varmathbb{R}}}  
\newcommand{\VS}[0]{\ensuremath{\varmathbb{V}}}  

\newcommand{\M}[0]{\ensuremath{\boldsymbol{\mathtt{M}}}}  
\newcommand{\Nn}[0]{\ensuremath{\boldsymbol{\mathtt{N}}}}  
\newcommand{\BSigma}[0]{\ensuremath{\boldsymbol{\Sigma}}}  
\newcommand{\X}[0]{\ensuremath{\boldsymbol{\mathtt{X}}}}  
\newcommand{\Y}[0]{\ensuremath{\boldsymbol{\mathtt{Y}}}}  
\newcommand{\Z}[0]{\ensuremath{\boldsymbol{\mathtt{Z}}}}  
\newcommand{\E}[0]{\ensuremath{\boldsymbol{\mathtt{E}}}}  
\newcommand{\JE}[0]{\ensuremath{\boldsymbol{\mathtt{J}}\E}}  
\newcommand{\JY}[0]{\ensuremath{\boldsymbol{\mathtt{J}}\Y}}  
\newcommand{\dJY}[0]{\ensuremath{\mathtt{J}\Y}}  
\newcommand{\dJZ}[0]{\ensuremath{\mathtt{J}\Z}}  
\newcommand{\Phase}[0]{\ensuremath{\boldsymbol{\mathtt{P}}}}  
\newcommand{\Fibre}[0]{\ensuremath{\boldsymbol{\mathtt{F}}}}  
\newcommand{\p}[0]{\ensuremath{\boldsymbol{\mathtt{p}}}}  
\newcommand{\q}[0]{\ensuremath{\boldsymbol{\mathtt{q}}}}  
\newcommand{\h}[0]{\ensuremath{\boldsymbol{\mathtt{h}}}}  
\newcommand{\djet}[0]{\ensuremath{\mathtt{j}}}  
\newcommand{\e}[0]{\ensuremath{\boldsymbol{\mathtt{e}}}}  
\newcommand{\GGG}[0]{\ensuremath{\boldsymbol{\mathtt{G}}}}  
\newcommand{\AAA}[0]{\ensuremath{\boldsymbol{\mathtt{A}}}}  
\newcommand{\AAAA}[2]{\ensuremath{\boldsymbol{\mathtt{A}}^{#1}_{\phantom{#1}#2}}}  
\newcommand{\BBB}[0]{\ensuremath{\boldsymbol{\mathtt{B}}}}  
\newcommand{\g}[0]{\ensuremath{\boldsymbol{\mathtt{g}}}}  
\newcommand{\EE}[0]{\ensuremath{\boldsymbol{\mathtt{E}}}}  
\newcommand{\bleps}{\ensuremath{\boldsymbol{\varepsilon}}} 

\newcommand{\dd}{\ensuremath{\boldsymbol{\mathtt{d}}}}  
\newcommand{\DD}{\ensuremath{\boldsymbol{\mathtt{D}}}}  
\newcommand{\lder}[1]{\ensuremath{\pounds}_{#1}}  
\newcommand{\Tan}[0]{\ensuremath{\mathtt{T}}}  
\newcommand{\Cotan}[0]{\ensuremath{\mathtt{T^*}}}  
\newcommand{\Ver}[0]{\ensuremath{\mathtt{V}}}  
\newcommand{\MVer}[0]{\ensuremath{\M\text{-}\Ver}}  
\newcommand{\HVer}[0]{\ensuremath{\HH\text{-}\Ver}}  
\newcommand{\Mhor}[0]{\ensuremath{\M\text{-}\mathtt{hor}}}  
\newcommand{\Hhor}[0]{\ensuremath{\HH\text{-}\mathtt{hor}}}  

\newcommand{\xxx}{\ensuremath{\boldsymbol{\mathtt{x}}}}	
\newcommand{\y}{\ensuremath{\boldsymbol{\mathtt{y}}}}	
\newcommand{\z}{\ensuremath{\boldsymbol{\mathtt{z}}}}	
\newcommand{\uu}{\ensuremath{\boldsymbol{\mathtt{u}}}}	
\newcommand{\vv}{\ensuremath{\boldsymbol{\mathtt{v}}}}	
\newcommand{\ww}{\ensuremath{\boldsymbol{\mathtt{w}}}}	
\newcommand{\VV}{\ensuremath{\boldsymbol{\mathtt{V}}}}	
\newcommand{\WW}{\ensuremath{\boldsymbol{\mathtt{W}}}}	
\newcommand{\bxi}{\ensuremath{\boldsymbol{\mathtt{\xi}}}}	
\newcommand{\bmu}{\ensuremath{\boldsymbol{\mathtt{\mu}}}}	
\newcommand{\bnu}{\ensuremath{\boldsymbol{\mathtt{\nu}}}}	
\newcommand{\Bxi}{\ensuremath{\boldsymbol{\mathtt{\Xi}}}}	
\newcommand{\bpi}{\ensuremath{\boldsymbol{\mathtt{\pi}}}}	
\newcommand{\Sec}{\ensuremath{\mathtt{Sec}}}	
\newcommand{\sect}{\ensuremath{\boldsymbol{\varphi}}}	
\newcommand{\forms}{\ensuremath{\mathtt{\Lambda}}}	
\newcommand{\mef}{\ensuremath{\mathtt{\Theta}}} 
\newcommand{\msf}{\ensuremath{\mathtt{\Omega}}}	
\newcommand{\kmsf}{\ensuremath{\boldsymbol{\mathtt{\omega}}}}	
\newcommand{\kmef}{\ensuremath{\boldsymbol{\mathtt{\theta}}}}

\newcommand{\blomega}{\ensuremath{\boldsymbol{\omega}}}	
\newcommand{\lsn}{\ensuremath{\left[\!\left[}}	
\newcommand{\rsn}{\ensuremath{\right]\!\right]}}	
\newcommand{\lpb}{\ensuremath{ \{\![}}	
\newcommand{\rpb}{\ensuremath{]\!\}}}	
\newcommand{\Lpb}{\ensuremath{ \big\{\!\big[}}	
\newcommand{\Rpb}{\ensuremath{\big]\!\big\}}}	
\newcommand{\Ham}[0]{\ensuremath{\mathtt{Ham}}}  
\newcommand{\TTT}[0]{\ensuremath{\boldsymbol{\mathtt{T}}}}  
\newcommand{\ttt}[0]{\ensuremath{\boldsymbol{\mathtt{t}}}}  
\newcommand{\rr}[0]{\ensuremath{\boldsymbol{\mathtt{r}}}}  
\newcommand{\sss}{\ensuremath{\boldsymbol{\mathfrak{s}}}}  

\newcommand{\f}[0]{\ensuremath{\boldsymbol{\mathtt{f}}}}  

\newcommand{\Aut}[1]{\ensuremath{\mathfrak{Aut}(#1)}}		
\newcommand{\Hom}[2]{\ensuremath{\mathfrak{Hom}(#1,#2)}}		
\newcommand{\GG}{\ensuremath{\mathtt{G}}}		
\newcommand{\alg}[1]{\ensuremath{\mathfrak{alg}(#1)}}		
\newcommand{\algx}[2]{\ensuremath{\mathfrak{alg}_{#2}(#1)}}		
\newcommand{\Diff}[1]{\ensuremath{\mathfrak{Diff}}(#1)}  
\newcommand{\dDiff}[1]{\ensuremath{\dd\text{-}\Diff{#1}}}  
\newcommand{\DiffM}[0]{\Diff{\M}}  
\newcommand{\SOf}[1]{\ensuremath{\mathtt{SO}(#1)}}  
\newcommand{\SO}[0]{\SOf{\bleta, \M}}  
\newcommand{\OO}[2]{\mathtt{O}^{#1}_{\phantom{#1}#2}}  
\newcommand{\GL}[0]{\ensuremath{\mathtt{GL}}}  

\newcommand{\proj}[2]{\ensuremath{\mathbbmss{p}_{#1\text{, }#2}}}  
\newcommand{\proh}{\ensuremath{\mathbbmss{p}}}  
\newcommand{\bleta}{\ensuremath{\etaup}}  
\newcommand{\bstar}[0]{\boldsymbol{*}}  

\begin{document}

\title{Covariant Quantum Gravity I:\\ Covariant Hamiltonian Framework}

\author{Mari\' an Pilc}


\institute{Institute of Theoretical Physics, Faculty of Mathematics and Physics,
Charles University in Prague, V Hole\v{s}ovi\v{c}k\'{a}ch 2, 180~00 Prague
8, Czech Republic\\
\email{marian.pilc@gmail.com}
}

\date{Received: date / Accepted: date}

\maketitle

\begin{abstract}
The first part of the series formulates the Einstein-Cartan theory in the covariant hamiltonian framework. The first section revises the general multisymplectic approach and introduces the notion of the d-jet bundles. Since the whole Standard Model Lagrangian (including gravity) can be written as the functional of the forms, the structure of the d-jet bundles is more appropriate for the covariant hamiltonian analysis than the standard jet bundle approach. The definition of the local covariant Poisson bracket on the space of covariant observables is recalled. The main goal of the work is to show that the gauge group of the Einstein-Cartan theory is given by the semidirect product of the local Lorentz group and the group of spacetime diffeomorphisms. Vanishing of the integral generators of the gauge group is equivalent to equations of motion of the Einstein-Cartan theory and the local covariant algebra generated by Noether's currents is closed Lie algebra.
\end{abstract}

\section{Introduction}

The first part of the series, which proposes as hypothesis a new theory of Covariant Quantum Gravity (CQG) with continuous quantum geometry, formulates the Einstein-Cartan theory within the covariant hamiltonian framework. The Einstein-Cartan theory, also familiar as the Kibble-Sciama theory, is a gauge theory where the local Poincar\' e group plays a role of the gauge symmetry\cite{Kibble,Hehl1,Hehl2}. Standard ADM formulation\cite{Wald} of General Relativity requires time+space splitting of the spacetime therefore the Hamilton formalism, which is necessary for any rigorous quantum formulation, breaks the explicit covariance and the algebra of constraints is no more closed Lie algebra, similar result can be obtained within the Einstein-Cartan theory\cite{Nikolic}. In the case when the gravitation is interacting with the pressure-free dust then there exists the privileged system of the coordinates comoving with every single grain of the dust which enables to rewrite equivalently the ADM constraints in such a way that they form closed Lie algebra\cite{BrownKuchar}. Kucha\v r also tried to rewrite the ADM constraints for the gravitation interacting with scalar field but in this case the ADM constraints can be rewritten only implicitly and the result depends on the solution of the equations of motion\cite{Kuchar}. Another possible solution of this Lie algebra problem was proposed in the Phoenix Project where all constraints of the system are contained in a single Master Constraint\cite{Thiemann_Master_Constraint}, but the Master Constraint for LQG is quadratic in the Hamiltonian constraint and given Lie algebra is no more associated with the local Poincar\' e group.\\
\hs The problem just mentioned yields the question whether there exists covariant hamiltonian formalism which can be applied here and which does not require the space+time splitting. Fortunatelly, we know that there exists an affirmative answer based on ideas of the multisymplectic geometry which generalize familiar symplectic structures\cite{Dedecker,GMS,Momentum_maps,HK1}.\\
\hs Usually, the covariant hamiltonian description works  with a notion of the jet bundles. The jet bundle is a fibre bundle constructed from the given configuration bundle $\Y$ over the spacetime $\M$, with the local coordinates $(x^{\mu},y^A)$, by adding the first derivatives of the variables $y^A$ to $\Y$, i.e. locally $(x^{\mu},y^A, v^A_{\mu})$. The Standard Model Lagrangian (including gravity) can be written as a functional of the forms over the spacetime $\M$ with values in a certain vector space (or in its submanifold as in the case of gravity) and their gauge covariant exterior derivatives, therefore it is more suitable to work within this structure, called the d-jet bundles, instead of the jet bundles. Therefore the basic results (covariant hamiltonian equations of motion, momentum maps, Noether's charges) in the language of the d-jet bundles should be introduced. This is the task of the section \ref{d_jet}.\\
\hs The section \ref{LGPB} deals with the construction of the local covariant Poisson bracket. We recall basic definitions of the local covariant bracket, observables and associated hamiltonian vector fields \cite{FPR} for the general multisymplectic manifold. In the section \ref{d_jet}  we have introduced two multisymplectic structures, the kinematical and dynamical, hence we need to know how their local covariant brackets are related. We also need to explore the general shape of the local observable, which is given by sum of generators of the group $\DiffM$ of all spacetime diffeomorphisms and $\M$-horizontal simply differentiable $(n-1)$-forms, see theorem \ref{Observable_canonical_decomposition_theorem}. Thus on general level we arrive into Lie algebra of local covariant observables. But as we know, the standard quantization procedure is based on searching of representations of different kind of an integral-like observables and Poisson bracket. We do not proceed such construction here, but rather we left it to the second part of the series where we explore instantaneous formalism\cite{Momentum_maps_II} in detail. On the other hand, the searching of the local Poisson algebra representation is also considered and explored in the literature and this step is used to be called a pre-quantization\cite{Kan3,Kan4}, but in some sence it can be viewed as the first quantization.\\
\hs We deal with the Einstein-Cartan theory in the section \ref{section_ECT}. At first we introduce a graded bundle of the right-handed coframes whose elements are interpreted as orthonormal vierbeins in the Einstein-Cartan theory. As another independent variable of the Einstein-Cartan theory is considered a metric-compatible connection and we finally arrive at the full covariant configuration bundle and the multisymplectic structure over it. Since the covariant Legendre map is singular we must proceed a multisymplectic reduction. Next, we explore equations of motion in the point of view of the covariant hamiltonian formalism. As the last thing, we find out that the gauge group of the Einstein-Cartan theory is given by the semidirect product of the local Lorentz group $\SO$ and the group of spacetime diffeomorphisms $\DiffM$ and show that equations of motion are given by vanishing of Noether's charges related to generators of the gauge group.\\
\\
NOTATION:\\\\
We use following convention in the series.\\\\
$\M - \text{Spacetime manifold with }\dim\M=n,$ (if we are dealing with Einstein-Cartan theory we have $n=4$).\\
$\BSigma - \text{Spatial manifold with }\dim\BSigma=n-1.$\\
\\
Spacetime indices are labeled by $\mu,\nu,\bar{\mu},\hat{\mu},\dots=0,1,\dots,n-1$, where indices with hats, bars are considered as standard without any additional meaning. Spatial indices are labeled by $\alpha,\beta,\gamma,\bar{\alpha},\hat{\alpha},\dots=1,\dots,n-1$.\\
\\
Multi-index notation\\\\
Let $B_{\mu_{p+1}\dots\mu_q}$ be totally antisymmetric $(0\le p<q\le n)$ then we set\\
$B_{(\mu)^p_q}=B_{\mu_{p+1}\dots\mu_q}$,\\
$B_{(\mu)_q}=B_{(\mu)^0_q}$,\\
$B_{(\mu)^p}=B_{(\mu)^p_n}$,\\
$B_{(\mu)}=B_{(\mu)^0_n}$.\\\\
Silent multi-index summation\\\\
$B_{(\mu)^p_q}C^{(\mu)^p_q}=\frac{1}{(q-p)!}B_{\mu_{p+1}\dots\mu_q}C^{\mu_{p+1}\dots\mu_q}.$\\\\
Lebesgue's coordinate measures\\\\
$\dd x^{(\mu)^p_q}=\dd x^{\mu_{p+1}}\wedge\dots \dd x^{\mu_q}$,\\
$\dd\Sigma_{(\mu)_q}=\varepsilon_{(\mu)_q(\mu)^q}\dd x^{(\mu)^q}$,\\
where $\varepsilon_{(\mu)}$ and $\bar{\varepsilon}_{(\mu)}$ are Levi-Civita symbols ($\varepsilon_{01\dots(n-1)}=\bar{\varepsilon}^{01\dots(n-1)}=1$).\\\\
Coframe indices runs through $a,b,\bar{a},...=0,1,...,n-1$ and Levi-Civita symbols $\bleps_{(a)}$ and $\bar{\bleps}^{(a)}$ are given by ($\bleps_{01\dots(n-1)}=\bar{\bleps}^{01\dots(n-1)}=1$) and Minkowski metric tensor has signature $(\bleta_{ab})=\text{diag}(-1,+1,+1,\dots,+1,+1)$. General multi-indices $A,B,\bar{A},\dots$ are running through some certain finite set depending on considered theory.

\section{Covariant Hamiltonian Formalism: d-jet bundles, Multisymplectic manifolds, Covariant momentum maps and Equations of Motion}\label{d_jet}
In the classical mechanics the hamiltonian analysis takes place on the symplectic manifold ($\Phase=\Cotan\Y,\omegaup$), where $\Y$ is the finite dimensional configuration space and $\omegaup=-\dd\varthetaup=-\dd (p_A\dd y^A)$ is the canonical symplectic $2$-form on $\Phase$. In the case of the classical field theory one usually starts with the infinite dimensional configuration manifold and then canonically constructs the infinite dimensional phase space. This construction requires the time+space splitting of the variables and the Legendre map relates the velocities, i.e. time derivatives of variables, with the canonical momenta. This breaks an explicit covariance. There also exists another approach\cite{Dedecker,Momentum_maps,GMS} which works with the multisymplectic structure and even more these two constructions are mutually complementary as we will see in the next part of the series. In this approach the infinite-dimensional configuration (or phase) space is replaced by the set of all (sufficiently smooth) sections of a certain finite-dimensional fibre bundle over the spacetime. While in the symplectic case the Legendre map gives, in the non-degenerate case, one-to-one relation between the velocities and the canonical momenta, the multisymplectic canonical momenta are related by the generalized or covariant Legendre map to the exterior derivatives of the fields.\\
\hs Let $\M$ be a spacetime manifold with dimension $n$ with local coordinates be $x^{\mu}$, where $\mu,\nu,\dots\in\{0,1,\dots,n-1\}$. Since we are interested in gravity where the metric is one of the observables it is also assumed that $\M$ is the topological smooth manifold only. It is well known\cite{Hawking_Ellis,Wald} that the equations of motion of the standard fields are well possessed if there exists global Cauchy surface and the spacetime has a structure $\M\simeq\BSigma\times\R$, where $x^0\in\R$ is interpreted as the time and $\BSigma_t=\{\xxx\in\M;x^0(\xxx)=t\}$ are supposed to be the achronal sections through $\M$ and play role of the Cauchy surfaces. As it was mentioned in the case of gravity there is no background metric and therefore the notion of the Cauchy surface depends on the solution. But if one assumes only the product condition $\M\simeq\BSigma\times\R$ and if for the initial conditions the initial embedding $\BSigma_{t_{\text{ini}}}$ of $\BSigma$ is the Cauchy surface then the equations of motion for the gravitational field are well possessed\cite{Wald,Hawking_Ellis}. Thus, let $\M\simeq\BSigma\times\R$. In order to avoid an analysis of the boundary terms and the overlapping conditions it is also assumed that $\BSigma$ is compact boundaryless orientable $(n-1)$-dimensional manifold and the considered fields are globally defined, therefore the bundles constructed in this part are trivial. Surprisingly, the triviality condition of the configuration bundle does not restrict our approach in $3+1$ dimensional space-time, since each $3$-dimensional manifold $\BSigma$ is parallelizable and therefore also $\M=\BSigma\times\R$ is also parallelizable.  Anyway, the notion of $d$-jet bundles and their duals is formulated for general case.\\
\hs The standard multisymplectic methods work with "scalar" fields $y^A$ and their spacetime derivatives $y^A_{,\,\mu}$. This means that if one wants to describe for example an electromagnetic field then $y^A\equiv A_{\mu}$ and the coordinate index $\mu$ is hidden in the general multi-index $A$ of $y^A$. On the other hand the Einstein-Hilbert-Palatini and the Standard model Lagrangians can be formulated within forms ($0$-forms: scalar, Dirac, neutrino fields or $1$-forms: electroweak, strong, gravitational connections or orthonormal coframe field) and their gauge-covariant exterior derivatives. Hence, it seems to be more suitable to develop the multisymplectic methods using directly the structures of the exterior algebra over the spacetime $\M$.
\subsection{Graded Manifolds}
By graded manifold $\Y$ we mean following (we are introducing the simplified definition rather than its general version in terms of categories, sheaves).
\begin{definition}
	Let $\proj{\M}{\Y}:\Y\to\M$ be a fibre bundle, where $\M,\Y$ are oriented open or closed finite dimensional smooth manifolds and the bundle projection $\proj{\M}{\Y}$ is smooth map. Let there exist a vector bundle $\M\VS$ given by direct sum 
	\be
	\M\VS = \opls\limits_{i=0}^n \forms^i(\M,\VS_i)
	\ee
	of vector bundles $\forms^i(\M,\VS_i)$ of $i$-forms over $\M$ with values in real finite dimensional vector space $\VS_i$ such that the typical fibre $\Y^f$ of the bundle $\Y$ is submanifold of the typical fibre vector space $\M\VS^f$ of the bundle $\M\VS$ and \be\Tan\Y^f=\Tan_{\Y^f}\M\VS^f\ee holds then $\Y$ is called graded manifold.
\end{definition}
NOTE: $\Y$ might be considered as complex manifold but for $\forall\y\in\Y$ we require $\y^{\bstar}\in\Y_{\xxx}$, where $\xxx=\proj{\M}{\Y}(\y)$ and $\Y_{\xxx}=\proj{\M}{\Y}^{-1}(\xxx)$ is the fibre over $\xxx$ containing $\y$ and ${}^{\bstar}$ is a complex conjugation inherited from the complex vector spaces $\VS_i$.\\
\hs The coordinates $\y^A$ on the fibre $\Y_{\xxx}$ are forms in $\xxx$ with values in a certain real vector space (or its submanifold as in the case of gravity) labeled by the index $A$ which is independent on the $\M$-coordinate indices $\mu,\nu,...$. In general, the degree $q_A$ of $\y^A$ may depends on $A$, but if there is no confusion we do not write this dependence explicitly $q=q_A$, i.e.
\be
\y^A=\frac{1}{q!}y^A_{\mu_1\dots\mu_q}\dd x^{\mu_1}\wedge\dots\wedge\dd x^{\mu_q}=y^A_{(\mu)_q}\dd x^{(\mu)_q}
\ee
keeping in mind this assumption.\\
\hs Let $\F$ be general smooth function $\F:\Y\to\forms^m\M$ satisfying $\proj{\M}{\forms\M}\circ\F=\proj{\M}{\Y}$. We say that $\F$ is simply differentiable if its variation given by small vertical changes $\delta\y$ can be written as
\ben
\F(\y+\delta\y)=\F(\y)+\delta\y^A\wedge\frac{\partial^L \F}{\partial\y^A}=\F(\vv)+\frac{\partial^R \F}{\partial\y^A}\wedge\delta\y^A.
\label{definition_simple_derivative}
\een
$(m-q)$-forms $\frac{\partial^L \F}{\partial\y^A}$ or $\frac{\partial^R \F}{\partial\y^A}$ are called left or right simple derivatives of $\F$ with respect to $\y^A$, respectively. These forms are nontrivial only for $0\le q\le m\le n$ and they are related by
\be
\frac{\partial^R\F}{\partial\;\y^A}=(-1)^{q(m-q)}\frac{\partial^L\F}{\partial\;\y^A}.
\ee
Following theorem shows the most general shape of simply differentiable function.

\begin{theorem}\label{Theorem_on_simple_differentiability}
	Smooth function $\F:\Y\to\forms^m\M$ satisfying $\proj{\M}{\forms\M}\circ\F=\proj{\M}{\Y}$ is simply differentiable if and only if one of the following conditions is satisfied
	\be
	\begin{tabular}{r p{9.5cm}}
		$\phantom{(}i)$& $m=n$,\\
		$\phantom{(}ii)$& $m<n$ and $\F$ is a finite polynom of the type
		\be
		\F=\F_0+\F_A\wedge\z^A+\frac{1}{2}\F_{AB}\wedge\z^A\wedge\z^B+\frac{1}{3!}\F_{ABC}\wedge\z^A\wedge\z^B\wedge\z^C+\dots,
		\ee
		where $\lambda^A$ is $0$-form part of $\y^A$, $\z^A$ is $q_A$-forms part of $\y^A$ with $0<q_A\le m$, $\ww^A$ is $q_A$-forms part of $\y^A$ with $m<q_A\le n$, i.e. $\y^A=\lambda^A+\z^A+\ww^A$, and smooth functions $\F_0$, $\F_A$, $\F_{AB}$, $\F_{ABC}$, $\dots$ depend only on $\xxx\in\M$ and $\lambda^A$.
		\\
	\end{tabular}
	\ee
\end{theorem}
\begin{proof}$\phantom{}$\\
	$\phantom{(i}i)$ $m=n$\\
	We have
	\be
	\F(\y^A+\delta\y^A)&=&\tilde{F}(y^A_{(\mu)_q}+\delta y^A_{(\mu)_q})\dd\Sigma=\tilde{F}\dd\Sigma+
	\frac{\partial\tilde{F}}{\partial y^A_{(\mu)_q}}\delta y^A_{(\mu)_q}\dd\Sigma,
	\ee
	where $\dd\Sigma$ is coordinate Lebesgue's volume form. Since in general
	\ben
	\delta y^A_{(\mu)_q}\dd\Sigma=\delta\y^A\wedge\dd\Sigma_{(\mu)_q} \label{n_form_times_components_of_r_form_formula}
	\een
	it is obvious that such $\F$ is always simply differentiable.\\
	$\phantom{(}ii)$ $m<n$\\
	We split $\y^A$ into three parts $\y^A=\lambda^A+\z^A+\ww^A$, where $\lambda^A$ is $0$-form part of $\y^A$, $\z^A$ is $q_A$-forms part of $\y^A$ with $0<q_A\le m$ and $\ww^A$ is $q_A$-forms part of $\y^A$ with $m<q_A\le n$. Within the proof we do not use Einstein summation convention for indices $A,B,...$ and we write this summation explicitly.\\
	For $\delta\y^A=\delta\ww^A$ we have 
	\be
	\frac{\partial^L \F}{\partial\ww^A}=0 \text{, hence }0=\delta\F=\sum\limits_A\frac{\partial\F}{\partial w^A_{(\mu)_{q_A}}}\delta w^A_{(\mu)_{q_A}}\text{, i.e. }\frac{\partial\F}{\partial w^A_{(\mu)_{q_A}}}=0.
	\ee
	For $\delta\y^A=\delta\lambda^A$ we have 
	\be
	\delta\F=\sum\limits_A\frac{\partial\F}{\partial \lambda^A}\delta\lambda^A=\sum\limits_A\frac{\partial\F}{\partial \lambda^A}\wedge\delta\lambda^A=\sum\limits_A\frac{\partial^R\F}{\partial \lambda^A}\wedge\delta\lambda^A.
	\ee
	For $\delta\y^A=\delta\ww^A$ the proof is little bit longer, therefore we do not proceed all steps explicitly, but rather we show only the main milestones of the proof. After some math the definition of simple differentiability (\ref{definition_simple_derivative}) for $\F=F_{(\nu)_m}\dd x^{(\nu)_m}$ yields
	\ben
	\frac{\partial F_{(\nu)_m}}{\partial z^A_{(\lambda)_{q_A}}}=\frac{m!}{\binom{n-m+q_A}{q_A}}\frac{\partial F_{(\mu)_{q_A}(\lambda)^{q_A}_m}}{\partial z^A_{(\mu)_{q_A}}}
	\delta^{(\lambda)_{q_A}(\lambda)^{q_A}_m}_{(\nu)_m}.\label{AUX_simple_differentiability_theorem_1}
	\een
	Multi-index $(\mu)_q$ can be immersed into the set of mutually different numbers $\{\mu_i\}_{i=1}^q=\{\mu_1,\dots,\mu_q\}$ therefore we can compare two multi-indices $(\mu)_q$ and $(\nu)_r$ by
	$(\mu)_q\subset(\nu)_r$ $\Leftrightarrow$ $\{\mu_i\}_{i=1}^q\subset \{\nu_j\}_{j=1}^r$ and define equivalence by $(\mu)_q\sim(\nu)_r$ $\Leftrightarrow$ $(\mu)_q\subset(\nu)_r$ and $(\nu)_r\subset(\mu)_q$. Relation (\ref{AUX_simple_differentiability_theorem_1}) implies
	\be
	\frac{\partial F_{(\nu)_m}}{\partial z^A_{(\lambda)_{q_A}}}=0
	\ee
	for $(\lambda)_{q_A}\not\subset(\nu)_m$ which means that $F_{(\nu)_m}$ depends only on those $z^A_{(\lambda)_{q_A}}$ having $(\lambda)_{q_A}\subset(\nu)_m$. Further, we have
	\be
	\frac{\partial F_{(\nu)_m}}{\partial z^A_{(\lambda)_{q_A}}}=\frac{m!}{\binom{n-m+q_A}{q_A}}\left(
	\xcancel{\sum\limits_{(\lambda)_{q_A}}}\frac{\partial F_{(\lambda)_{q_A}(\lambda)^{q_A}_m}}{z^A_{(\lambda)_{q_A}}}+
	\sum\limits_{(\mu)_{q_A}\nsim(\lambda)_{q_A}}\frac{\partial F_{(\mu)_{q_A}(\lambda)^{q_A}_m}}{z^A_{(\mu)_{q_A}}}
	\right)\delta^{(\lambda)_{q_A}(\lambda)^{q_A}_m}_{(\nu)_m}.
	\ee
	The canceled summation mark in the first term means that we do not sum through the repeating indices $(\lambda)_{q_A}$. The second term in big brackets depends on $z^A_{(\mu)_{q_A}}$ having $(\mu)_{q_A}\nsim(\lambda)_{q_A}$for $(\lambda)_{q_A}(\lambda)^{q_A}_m\subset(\nu)_m$, while the left hand side of the equation does not. This is possible only if $\frac{\partial F_{(\nu)_m}}{\partial z^A_{(\lambda)_{q_A}}}$ does not depend on $z^A_{(\lambda)_{q_A}}$, i.e. $F_{(\nu)_m}$ depends on $z^A_{(\lambda)_{q_A}}$ having $(\lambda)_{q_A}\subset(\nu)_m$ at most linearly. This implies that $F_{(\nu)_m}$ is finite polynomial in $z^A_{(\nu)_{\q_A}}$
	\be
	F_{(\nu)_m}=\sum\limits_{r=0}^{\infty}\alpha^r_{(\nu)_m},
	\ee
	where, $q_i=q_{A_i}$,
	\ben
	\alpha^r_{(\nu)_m}=\frac{1}{r!}
	\sum_{A_1,\dots,A_r}
	\beta_{A_1\dots A_r\,(\nu)_m}^{\phantom{A_1\dots A_r\,}(\lambda^1)_{q_1}:\dots:(\lambda^r)_{q_r}}z^{A_1}_{(\lambda^1)_{q_1}}\dots z^{A_r}_{(\lambda^r)_{q_r}}\label{AUX_simple_differentiability_theorem_2}
	\een
	are homogeneous polynomials of the $r$-th order and only finite number of them are not vanishing identically. $\beta_{A_1\dots A_r\,(\nu)_m}^{\phantom{A_1\dots A_r\,}(\lambda^1)_{q_1}:\dots:(\lambda^r)_{q_r}}$'s depend on $\xxx\in\M$ and $\lambda^A$ only. The colon between two multi-indices $(\lambda^1)_{q_1}:(\lambda^2)_{q_2}$ means that these multi-indices cannot be exchanged, i.e. while two multi-indices without column are graded-symmetric $\Lambda_{(\lambda^1)_{q_1}(\lambda^2)_{q_2}}=(-1)^{q_1 q_2}\Lambda_{(\lambda^2)_{q_2}(\lambda^1)_{q_1}}$
	there is no such explicit relation between  $\Lambda_{(\lambda^1)_{q_1}:(\lambda^2)_{q_2}}$ and $\Lambda_{(\lambda^2)_{q_2}:(\lambda^1)_{q_1}}$. Since two homogeneous polynomials with different orders are linearly independent we can deal directly with homogeneous parts $F_{(\nu)_m}=\alpha^r_{(\nu)_m}$ without lost of generality.\\
	\hs Now, if we enter the ansatz (\ref{AUX_simple_differentiability_theorem_2}) into (\ref{AUX_simple_differentiability_theorem_1}) and use relation
	\be
	\alpha^r_{(\nu)_m}=\int\limits_{\gamma}\sum\limits_{A}\frac{\partial \alpha^r_{(\nu)_m} }{\partial z^A_{(\mu)_{q_A}}}\dd z^A_{(\mu)_{q_A}},
	\ee
	where the integration path is given by $\gamma(t)=(t\cdot z^A_{(\mu)_{q_A}})$ for $t\in\langle 0;1\rangle$, then we get after some math, $Q_r=q_1+\dots+q_r$,
	\ben
	\beta_{A_1\dots A_r\,(\nu^1)_{q_1}\phantom{:}\dots\phantom{:}(\nu^r)_{q_r}(\nu)^{Q_r}_m}^{\phantom{A_1\dots A_r\,}(\lambda^1)_{q_1}:\dots:(\lambda^r)_{q_r}}&=&
	\frac{(m-Q_r)!\left(\prod\limits_{i=1}^r q_i!\right)}{r}\times \label{AUX_simple_differentiability_theorem_3}\\
	&&\sum\limits_{i=1}^r\frac{\binom{m}{q_i}}{\binom{n-m+q_i}{q_i}}
	\beta_{A_1\dots A_r\,(\mu^1)_{q_1}\phantom{:}\dots\phantom{:}(\mu^i)_{q_i}\phantom{:}\dots\phantom{:}(\mu^r)_{q_r}(\mu)^{Q_r}_m}^{\phantom{A_1\dots A_r\,}(\lambda^1)_{q_1}:\dots:(\mu^i)_{q_i}:\dots:(\lambda^r)_{q_r}}
	\delta^{(\mu^1)_{q_1}\dots(\lambda^i)_{q_i}\dots(\mu^r)_{q_r}(\mu)^{Q_r}_m}_{(\nu^1)_{q_1}\dots(\nu^i)_{q_i}\dots(\nu^r)_{q_r}(\nu)^{Q_r}_m}
	\nonumber
	\een
	We can more specify coefficients $\beta_{A_1\dots A_r\,(\nu^1)_{q_1}\phantom{:}\dots\phantom{:}(\nu^r)_{q_r}(\nu)^{Q_r}_m}^{\phantom{A_1\dots A_r\,}(\lambda^1)_{q_1}:\dots:(\lambda^r)_{q_r}}$ if we realize that
	\be
	\xcancel{\sum\limits_{(\lambda^1)_{q_1}\dots(\lambda^r)_{q_r}}}
	\frac{\partial^r \alpha^r_{(\lambda^1)_{q_1}\dots(\lambda^r)_{q_r}(\lambda)^{Q_r}_m}}{\partial z^{A_1}_{(\lambda^1)_{q_1}}\dots \partial z^{A_r}_{(\lambda^r)_{q_r}}}=
	\xcancel{\sum\limits_{(\lambda^1)_{q_1}\dots(\lambda^r)_{q_r}}}
	\beta_{A_1\dots A_r\,(\lambda^1)_{q_1}\phantom{:}\dots\phantom{:}(\lambda^r)_{q_r}(\lambda)^{Q_r}_m}^{\phantom{A_1\dots A_r\,}(\lambda^1)_{q_1}:\dots:(\lambda^r)_{q_r}}=
	\Gamma_{A_1\dots A_r\,(\lambda)^{Q_r}_m}((\lambda)_m)	
	\ee	
	does not depend on $z^A_{(\lambda)_{q_A}}$. On the other hand at this moment we do not know whether $\Gamma$ term is same for all $(\lambda)_m$ or not, hence we write this dependence explicitly.
	We have (no silent multi-index summation!)
	\ben
	\beta_{A_1\dots A_r\,(\nu^1)_{q_1}\phantom{:}\dots\phantom{:}(\nu^r)_{q_r}(\nu)^{Q_r}_m}^{\phantom{A_1\dots A_r\,}(\lambda^1)_{q_1}:\dots:(\lambda^r)_{q_r}}=
	\frac{m!}{(Q_r-m)!}
	\sum\limits_{(\lambda)^{Q_r}_m}
	\delta_{(\nu^1)_{q_1}\dots(\nu^r)_{q_r}(\nu)^{Q_r}_m}^{(\lambda^1)_{q_1}\dots(\lambda^r)_{q_r}(\lambda)^{Q_r}_m}\Gamma_{A_1\dots A_r\,(\lambda)^{Q_r}_m}((\nu)_m).
	\label{AUX_simple_differentiability_theorem_4}
	\een
	Now, if we enter (\ref{AUX_simple_differentiability_theorem_4}) into (\ref{AUX_simple_differentiability_theorem_3}) we find out
	\be
	\Gamma_{A_1\dots A_r\,(\nu)^{Q_r}_m}((\nu)_m)=\Gamma_{A_1\dots A_r\,(\nu)^{Q_r}_m}.
	\ee
\qed
\end{proof}
\subsection{$d$-jet Bundles}
\hs Graded manifold $\Y$ plays role of the basic playground where all possible physical configurations or states are given by set of all possible finite families of local smooth sections $\{\sect_i\}_{i\in I}$, i.e. maps $\sect_i:\DD_{\sect_i}\subset\M\to\Y$, $\proj{\M}{\Y}\circ\sect_i=\id$, with suitable overlapping conditions on $\DD_{\sect_i}\cap\DD_{\sect_j}$ and covering definition ranges $\Union\limits_{i\in I}\DD_{\sect_i}=\M$. We denote the set of local sections by $\Sec_0(\M,\Y)$ while $\Sec(\M,\Y)$ is used for the set of global sections,i.e. set of those mappings $\sect\in\Sec_0(\M\,\Y)$ having definition range $\DD_{\sect}$ equal to $\M$. Local section $\sect\in\Sec_0(\M,\Y)$ defines a submanifold $\sect(\DD_{\sect})\subset\Y$ given locally by $\big(x^{\mu}(\xxx),\y^A_{\sect}(\xxx)\big)$. Two sections $\sect$, $\sect '\in \Sec_0(\M,\Y)$ are equivalent $\sect\simx{\xxx}\sect '$ in the point $\xxx\in\DD_{\sect}\cap\DD_{\sect '}\subset\M$ if $\sect(\xxx)=\sect '(\xxx)$, i.e. $\big(x^{\mu}(\xxx),\y^{A}_{\sect}(\xxx)\big)=\big(x^{\mu}(\xxx),\y^{A}_{\sect '}(\xxx)\big)$, and 
$\dd\y^{A}_{\sect}(\xxx)=\dd\y^{A}_{\sect '}(\xxx)$, where $\dd$ is the exterior derivative on $\M$. If the equivalence class of $\sect(\xxx)$ is denoted $\big[\sect\big]_{\xxx}$ then the set defined by
\be
\dJY:=\bigcup_{\substack{\xxx\in\M\\\sect\in\Sec_0(\M,\Y)}}\big[\sect\big]_{\xxx}
\ee
is the manifold which can be locally described by coordinates $(x^{\mu}, \y^A,\vv^A)$ where $\vv^A$, the generalized velocities of $\y^A$, are $(q+1)$-forms on $\M$. The projection $\proj{\Y}{\dJY}:\dJY\to\Y$, given by $\proj{\Y}{\dJY}:(x^{\mu}, \y^A,\vv^A)\mapsto(x^{\mu}, \y^A)$, turns $\dJY$ into the fibre bundle and the chain $ \dJY\stackrel{\proj{\Y}{\dJY}}{\longrightarrow}\Y\stackrel{\proj{\M}{\Y}}{\longrightarrow}\M$ is called d-jet triple.\\
\hs Let $\sect\in\Sec_0(\M,\Y)$ then $\djet\sect:\M\to\dJY$ defined by 
\be\djet\sect:\xxx\mapsto\big[\sect\big]_{\xxx}=\big(x^{\mu}(\xxx),\y^A_{\sect}(\xxx),\dd\y^A_{\sect}(\xxx)\big)\ee
is a section on the bundle $\proj{\M}{\Y}\circ\proj{\Y}{\dJY}=\proj{\M}{\dJY}:\dJY\to\M$ called the d-jet prolongation of $\sect$. Section of the bundle $\proj{\M}{\JY}:\JY\to\M$ is called holonomic if it is d-jet prolongation of some section of the bundle $\proj{\M}{\Y}:\Y\to\M$. If for $\forall A$ we have $q=0$ then d-jet bundle reduces into the first jet bundle. For $q=n$ the d-jet bundle is trivial, since every $\M$ horizontal $(n+1)$-form on $\M$ is trivially vanishing, thus $\dJY\simeq\Y$.\\
\hs Let $\dJY\to\Y\to\M$ and $\dJZ\to\Z\to\Nn$ be two d-jet triples and let there exists a homomorphism between bundles $\Y\to\M$ and $\Z\to\Nn$, i.e. two diffeomorphisms $\bleta_{\Nn\M}:\M\to\Nn$ and $\bleta_{\Z\Y}:\Y\to\Z$ tied by $\bleta_{\Nn\M}\circ\proj{\M}{\Y}=\proj{\Nn}{\Z}\circ\bleta_{\Z\Y}$. If $\sect\in\Sec_0(\M,\Y)$, then $\bleta_{\Z\Y}\circ\sect\circ(\bleta_{\Nn\M})^{-1}$ is a section on $\Z\to\Nn$. If for every point $\forall\xxx\in\M$ and for arbitrary two equivalent sections $\sect\simx{\xxx}\sect '\in\Sec_0(\M,\Y)$ the structural condition 
\ben
\bleta_{\Z\Y}\circ\sect\circ(\bleta_{\Nn\M})^{-1}\simx{\bleta_{\Nn\M}(\xxx)}\,\bleta_{\Z\Y}\circ\sect '\circ(\bleta_{\Nn\M})^{-1}\label{structural_condition}
\een
is satisfied, then it is possible to prolong the map $\bleta_{\Z\Y}$ to the d-jet part of the triples by
\be
\bleta_{\dJZ\dJY}\equiv\djet\bleta_{\Z\Y}:\big[\sect\big]_{\xxx}\mapsto\big[\bleta_{\Z\Y}\circ\sect\circ(\bleta_{\Nn\M})^{-1}\big]_{\bleta_{\Nn\M}(\xxx)}.
\ee
A tri-map $\bleta=\big( \bleta_{\Nn\M}, \bleta_{\Z\Y},\djet\bleta_{\Z\Y}\big):\big(\M, \Y,\JY\big)\to\big(\Nn, \Z,\dJZ\big)$ is called d-jet homomorphism.
 We write $\bleta\in\Hom{\Y}{\Z}$, or simply $\bleta_{\Z\Y}\in\Hom{\Y}{\Z}$. In the case when two $d$-jet triples coincide then the tri-map $\bleta$ is called d-jet automorphism and set of all d-jet authomorphisms is denoted by $\Aut{\Y}$. In general, it may happen that the inverse maps $(\bleta^{-1}_{\Nn,\M},\bleta_{\Z\Y}^{-1})$ are not d-jet homomorphism, see example given by transformations (\ref{non_d_jet_transformation}) and (\ref{canonical_momenta_ECT}). In the case when the inverse maps $(\bleta^{-1}_{\Nn,\M},\bleta_{\Z\Y}^{-1})$ are also d-jet homomorphism then the tri-map $\bleta=\big( \bleta_{\Nn\M}, \bleta_{\Z\Y},\djet\bleta_{\Z\Y}\big)$ is called d-jet diffeomorphism and set of all d-jet diffeomorphisms is denoted by $\dDiff{\Y}$.\\

\hs Let $\bleta=\big( \bleta_{\Nn\M}, \bleta_{\Z\Y},\djet\bleta_{\Z\Y}\big)\in\Hom{\Y}{\Z}$ be d-jet homomorphism from $\Y$ to $\Z$ then by definition $\bleta_{\Z\Y}$ is diffeomorphism, i.e. infinitely many times differentiable map, but it does not mean that $\bleta_{\Z\Y}$ is (in)finitely many times simply differentiable. Following theorem says how general d-jet homomorphism can be characterised within the notion of the simple differentiation.
\begin{theorem}\label{Theorem_d_jet_homomorphism_extension}
	Bundle homomorphism $(\bleta_{\Nn\M},\bleta_{\Z\Y})$ can be extended into d-jet homomorphism if and only if
	\be
	&\bleta_{\Nn\M}:&\big(x^{\mu}\big)\mapsto\big(\bar{x}^{\bar{\mu}}=(\bleta_{\Nn\M})^{\bar{\mu}}(x^{\mu})\big),\\
	&\bleta_{\Z\Y}:&\big(x^{\mu},\y^A)\mapsto
	\Big(\bar{x}^{\bar{\mu}}=\big(\bleta_{\Nn\M})^{\bar{\mu}}(x^{\mu}\big), \bar{\y}^{\bar{A}}=\big(\bleta^{-1}_{\Nn\M}\big)^*\overline{\Y}^{\bar{A}}(x^{\mu},\y^A)\Big),
	\ee
	where each $\overline{\Y}^{\bar{A}}\in\Sec(\Y,\forms\Y)$ is simply differentiable $\M$-horizontal form on $\Y$.
\end{theorem}
\begin{proof}
	Let $(x^{\mu},\y^A=y^A_{(\mu)_q}\dd x^{(\mu)_q})$ or $(\bar{x}^{\bar{\mu}},\bar{\y}^{\bar{A}}=\bar{y}^{\bar{A}}_{(\bar{\mu})_{\bar{q}}}\dd \bar{x}^{(\bar{\mu})_{\bar{q}}})$ be local coordinates on bundle $\Y\to\M$ or $\Z\to\Nn$, respectively, then $(\bleta_{\Nn\M},\bleta_{\Z\Y})$ is bundle homomorphism if and only if it looks like
	\be
	&\bleta_{\Nn\M}:&\big(x^{\mu}\big)\mapsto\big((\bleta_{\Nn\M})^{\bar{\mu}}(x^{\mu})\big),\\
	&\bleta_{\Z\Y}:&\big(x^{\mu},y^A_{(\mu)_q})\mapsto
	\Big(\big(\bleta_{\Nn\M})^{\bar{\mu}}(x^{\mu}\big), \Z^{\bar{A}}\big(x^{\mu},y^A_{(\mu)_q}\big)\Big).
	\ee
	Now, let $\sect\in\Sec_0(\M,\Y)$ be arbitrary section given locally $\sect(\xxx)=\big(x^{\mu}(\xxx),\y^A_{\sect}(\xxx)\big)\equiv\Big(x^{\mu},y^A_{(\mu)_q}\big(x^{\mu}\big)\Big)$, then section on $\Z\to\Nn$ is given by
	\be
	\bleta_{\Z\Y}\circ\sect\circ(\bleta_{\Nn\M})^{-1}(\bar{\xxx})&=&
	\bigg(
	\bar{x}^{\bar{\mu}}(\bar{\xxx}), Z^{\bar{A}}_{(\bar{\mu})_{\bar{q}}}\Big(x^{\mu}(\bar{\xxx}),y^A_{(\mu)_q}\big(x^{\mu}(\bar{\xxx})\big)\Big)
	\bigg),\\
	&=&
	\big(
	\bar{x}^{\bar{\mu}}(\bar{\xxx}), \Z^{\bar{A}}_{\sect}(\bar{\xxx})\big),
	\ee
	where $x^{\mu}(\bar{\xxx})=\big(\bleta_{\Nn\M}^{-1}\big)^{\mu}(\bar{\xxx})$ for short. Now, if we set $\overline{\Y}^{\bar{A}}=(\bleta_{\Nn\M})^*\Z^{\bar{A}}$ and calculate $\dd\Z^{\bar{A}}_{\sect}$ then we get
	\be
	(\bleta_{\Nn\M})^*\dd\Z^{\bar{A}}_{\sect}=\dd\overline{\Y}^{\bar{A}}_{\sect}=
	\bar{Y}^{\bar{A}}_{(\lambda)_{\bar{q}},\nu}\dd x^{\nu}\wedge\dd x^{(\lambda)_{\bar{q}}}+
	\frac{\partial\bar{Y}^{\bar{A}}_{(\lambda)_{\bar{q}}}}{\partial y^A_{(\mu)_q}}y^A_{(\mu)_q,\nu}\dd x^{\nu}\wedge\dd x^{(\lambda)_{\bar{q}}}.
	\ee
	Structural condition (\ref{structural_condition}) yields that $\dd\Z^{\bar{A}}_{\sect}$ should depends only on antisymmetric part of $y^A_{(\mu)_q,\nu}$. Since index $\nu$ is summarized with $\dd x^{\nu}$ we have that the last term should looks like
	\be
	\frac{\partial\bar{Y}^{\bar{A}}_{(\lambda)_{\bar{q}}}}{\partial y^A_{(\mu)_q}}y^A_{(\mu)_q,\nu}\dd x^{\nu}\wedge\dd x^{(\lambda)_{\bar{q}}}
	=y^A_{(\mu)_q,\nu}\dd x^{\nu}\wedge\dd x^{(\mu)_q}\wedge Q^{\bar{A}}_{A(\lambda)^{\bar{q}}_q}\dd x^{(\lambda)^{\bar{q}}_q}.
	\ee
	This condition is satisfied if only if $q\le\bar{q}$ and
	\be
	\frac{\partial\bar{Y}^{\bar{A}}_{(\lambda)_{\bar{q}}}}{\partial y^A_{(\mu)_q}}=
	\binom{\bar{q}}{q}\delta^{\phantom{[}(\mu)_q\phantom{\bar{A}]}}_{[(\lambda)_q\phantom{A}}Q^{\bar{A}}_{A(\lambda)^{q}_{\bar{q}}]}
	\ee
	or $\frac{\partial\bar{Y}^{\bar{A}}_{(\lambda)_{\bar{q}}}}{\partial y^A_{(\mu)_q}}=0$. Using this formula in Taylor expansion of 
	$\bar{\Y}^{\bar{A}}$ yields immediately that $\bar{\Y}^{\bar{A}}$ is simply differentiable. Proof of opposite implication is straightforward, hence we leave it for reader.
	
\qed\end{proof}
\subsection{Duals of $d$-jet bundles}
Let us introduce a dual of the d-jet bundle. It plays an important role in the covariant hamiltonian formalism. It can be seen that $\vv^A$-part of coordinates on the d-jet bundle transforms for general d-jet diffeomorphism in affine way. In other words $\vv^A$  behaves as point in affine space rather than vector. So, if one wants to talk about dual of $\vv^A$-space then one must consider its affine dual. For the purposes of the covariant hamiltonian formalism it is the most suitable to consider affine maps to the space $\forms^n\M$ of $n$-forms on $\M$. Using (\ref{n_form_times_components_of_r_form_formula}) we see that these can be labeled by $(n-q-1)$-forms $\p_A$ and $n$-form $\h$ where the affine map $\aff$ is given by
\be
\aff(\vv^A)=(-1)^{n-q}\p_A\wedge\vv^A+\h.
\ee
The factor $(-1)^{n-q}$ is chosen in such way that the local covariant Poisson bracket constructed in the next section among canonically conjugated coordinates $\y^A,\p_B$ is equal to identity, see lemma \ref{Lemma_on_general_smeared_canonical_observable}. So, the dual $\dJY^*$ of the d-jet bundle $\dJY$ is defined as a fibre bundle $\proj{\Y}{\dJY^*}:\dJY^*\to\Y$ locally described by coordinates as
$(x^{\mu},\; \y^A,\; \p_A,\; \h)$.\\
\hs There can be constructed a couple of forms on the d-jet dual $\dJY^*$. The first canonical $n$-form, called the kinema\-tical Cartan-Poincar\' e form,  is defined as
\ben\label{mef_dual}
\kmef=(-1)^{n-q}\p_A\wedge\dd\y^A+\h
\een
and the second canonical $(n+1)$-form, called the kinematical multisymplectic form, is given by
\ben\label{msf_dual}
\kmsf=-\dd\kmef=(-1)^{n-q-1}\dd\p_A\wedge\dd\y^A-\dd\h.
\een
The word 'multisymplectic' means following. Let $(\Fibre,\kmsf)$ be a pair of manifold equipp\-ed with nondegenerate, i.e. mapping $\vv\in\Tan\Fibre\mapsto i_{\vv}\kmsf$ has only trivial kernel, closed $(n+1)$-form $\kmsf$ on $\Fibre$, $n<\dim \Fibre$, then the pair $(\Fibre,\kmsf)$ is called the multisymplectic manifold. We assume that $\kmsf$ is also exact, i.e. there exists the $n$-form $\mef$, such that $\kmsf=-\dd\mef$. It is obvious that if $n=1$ (spacetime is only time), then the definition reduces to the usual definition of the symplectic manifold. 
\subsection{Covariant Legendre Transformation}
In mechanics the Legendre transformation  relates the velocities with the canonical momenta and the Lagrangian with the Hamiltonian. Such transformation is driven by the Lagrange function. In the multisymplectic context it is similar. Lagrange function $\LL:(x^{\mu},\y^A,\vv^A)\mapsto\LL(x^{\mu},\y^A,\vv^A)\in\forms^n\M$ is a map from the d-jet bundle $\dJY$ to the space of $n$-forms on $\M$. Let $\bgm$, $\bgm_0\in\dJY_{\y}:=(\proj{\Y}{\dJY})^{-1}(\y)$ then the affine approximation of the Lagrangian in the point $\bgm_0$ on $\dJY_{\y}$
\be
\LL(\bgm)\simeq\LL(\bgm_0)+\frac{\partial^R\LL}{\partial\,\vv^A}(\bgm_0)\wedge(\ww^A-\vv^A),
\ee
where $\bgm=(x^{\mu},\y^A,\ww^A)$ and $\bgm_0=(x^{\mu},\y^A,\vv^A)$, defines an element of the d-jet dual $\dJY^*$ by
\be
-\LL(\bgm)\simeq-\LL-\frac{\partial^R\LL}{\partial\,\vv^A}\wedge(\ww^A-\vv^A)=(-1)^{n-q}\p_A\wedge\ww^A+\h.
\ee
So, the Legendre transformation $\LTr$ is a map from $\dJY$ to its d-jet dual $\dJY^*$ given by
\ben
\LTr:(x^{\mu},\y^A,\vv^A)\mapsto\bigg(x^{\mu},\y^A,\p_A=(-1)^{n-q-1}\frac{\partial^R\LL}{\partial\,\vv^A},
\h=\frac{\partial^R\LL}{\partial\,\vv^A}\wedge\vv^A-\LL\bigg).\label{Legendre_transformation}
\een
Let $\Phase\subset\dJY^*$ be a submanifold of the d-jet dual defined as an image of the Legendre transformation, i.e. $\Phase:=\LTr(\dJY)$, and let $\LTr_{\Phase}$ be $\LTr$ considered as a map onto $\Phase$. If $\LTr_{\Phase}:\dJY\to\Phase$ is a diffeomorphism then the Lagrangian is called regular and then there also exists its inverse $\LTr_{\Phase}^{-1}:\Phase\to\dJY$ which allows to express the generalized velocities $\vv^A$ as functions of $x^{\mu},\y^A,\p_A$. The fibre bundle $\proj{\Y}{\Phase}:\Phase\to\Y$ is called the dynamical phase bundle, or simply the phase bundle. The phase bundle $\Phase$ is locally described by the coordinates $(x^{\mu},\y^A,\p_A)$, where the coordinates $\p_A$ are called the canonical momenta. From now until the section \ref{section_ECT} we are dealing with regular Lagrangians only.\\
\hs There can be defined the canonical forms in regular case by using the canonical forms on $\dJY^*$ and the Legendre transformation. The (dynamical) Cartan-Poincar\' e $n$-form $\mef$ on $\Phase$ is 
\ben
\mef:=\kmef\vb{\Phase}=(-1)^{n-q}\p_A\wedge\dd\y^A+\HH,\label{mef_dynamical}
\een
where $\M$-horizontal map $\HH:\Phase\to\forms^n\Phase$ is the Hamiltonian defined by
\ben
\HH\circ\LTr=\frac{\partial^R\LL}{\partial\,\vv^A}\wedge\vv^A-\LL.\label{general_hamiltonian}
\een
The (dynamical) multisymplectic $(n+1)$-form is given in a standard way as
\ben
\msf:=-\dd\mef=(-1)^{n-q-1}\dd\p_A\wedge\dd\y^A-\dd\HH.\label{msf_dynamical}
\een

\subsection{Equations of Motion, Symmetry Group and Noether's Charges}
\hs The main goal of this subsection is to derive the covariant Hamilton equations of motion from the variational action principle, find the relation between them and the symmetry group of the system. This relation rarely exists, so the sufficient conditions should be derived. It will be shown in the next section that in the case of the Einstein-Cartan theory the conditions are satisfied. From now due to simplicity we suppose that the bundle $\Y\to\M$ is trivial. Therefore $\JY\to\M$, $\dJY^*\to\M$ and $\Phase\to\M$ are also trivial. Let $\sect\in\Sec(\M,\Phase)$ be the section of the phase bundle $\Phase\to\M$. Then the Lagrangian evaluated on the section $\sect$ is given by the pullback of the 
Cartan-Poincar\' e $n$-form $\mef$
\ben
\LL(x^{\mu},\y^A_{\sect},\p_A^{\sect},\dd\y^A_{\sect})=-\sect^*(\mef)=(-1)^{n-q-1}\p^{\sect}_A\wedge\dd\y^A_{\sect}-
\HH(x^{\mu},\y^A_{\sect},\p_A^{\sect}).\label{Lagrangian_pullback_of_CP_form}
\een
The action of the system in the state given by the section $\sect$ on the spacetime interval $\M_{\Ii}=\cup_{t\in \Ii}\BSigma_t$, where $\Ii=\langle t_{ini},\,t_{fin}\rangle$, is given by
\ben
\Saction_{\Ii}(\sect)=\int\limits_{\M_{\Ii}}-\sect^*(\mef)\label{action_integral}.
\een
There exists a following theorem\cite{Momentum_maps}
\begin{theorem}\label{theorem_eom_msf}
	Section $\sect\in\Sec(\M,\Phase)$ is a stationary point of the action integral (\ref{action_integral}) if and only if
	for all vector fields $\vv\in\Sec(\Phase,\Tan\Phase)$ on the phase bundle $\Phase$ the following expression is satisfied
	\ben
	\sect^*(i_{\vv}\msf)=0.\label{msf_vanishinig}
	\een
\end{theorem}

\begin{proof}
	A variation of the action integral (\ref{action_integral}) is, where $\sect+\deltaup\sect=\big(x^{\mu}(\xxx),\y^A_{\sect}(\xxx)+\deltaup\y^A(\xxx),\p^{\sect}_A(\xxx)+\deltaup\p_A(\xxx)\big)$ means small variation of the section $\sect$,
	\be
	\deltaup\Saction_{\Ii}(\sect)
	=\int\limits_{\M_{\Ii}}\left[
	\deltaup\p_A\wedge\sect^* \left((-1)^{n-q-1}\dd\y^A-\frac{\partial^L\HH}{\partial\,\p_A}\right)+
	\sect^*\left(-\,\dd\p_A-\frac{\partial^R\HH}{\partial\,\y^A}\right)\wedge\deltaup\y^A
	\right].
	\ee
	The stationarity condition of (\ref{action_integral}) yields the covariant Hamilton equations
	\ben
	0&=&\,\sect^* \,\left(\;(-1)^{n-q-1}\dd\y^A\,-\,\frac{\partial^L\HH}{\partial\,\p_A}\,\right)\;=\;
	(-1)^{n-q-1}\dd\y^A_{\sect}(\xxx)-\frac{\partial^L\HH}{\partial\,\p_A}\big(x^{\mu}(\xxx),\y^A_{\sect}(\xxx),\p^{\sect}_A(\xxx)\big),\nonumber\\
	\label{Hamilton_equations}\\
	0&=&\sect^*\left(-\,\dd\p_A-\frac{\partial^R\HH}{\partial\,\y^A}\right)=
	-\,\dd\p^{\sect}_A(\xxx)-\frac{\partial^R\HH}{\partial\,\y^A}\big(x^{\mu}(\xxx),\y^A_{\sect}(\xxx),\p^{\sect}_A(\xxx)\big).\nonumber
	\een
	Let $\vv\in\Sec(\Phase,\Tan\Phase)$ be arbitrary vector field on $\Phase$. It can be expressed as
	\be
	\vv=\xi^{\mu}\partial_{\mu}+v^A_{(\nu)_q}\partial_{y^A_{(\nu)_q}}+w_{A\,(\nu)^{q+1}}\partial_{p_{A\,(\nu)^{q+1}}}\equiv
	\bxi+\vv^A\partial_{\y^A}+\ww_A\partial_{\p_A}
	\ee
	Action of the interior product of $\vv$ with the dynamical multisymplectic form $\msf$ is given by
	\ben
	i_{\vv}\msf&=&\ww_A\wedge\left((-1)^{n-q-1}\dd\y^A-\frac{\partial^L\HH}{\partial\,\p_A}\right)-
	\left(-\,\dd\p_A-\frac{\partial^R\HH}{\partial\,\y^A}\right)\wedge\vv^A-\nonumber\\
	& &-i_{\bxi}\left\{
	\left(-\,\dd\p_A-\frac{\partial^R\HH}{\partial\,\y^A}\right)\wedge\left((-1)^{n-q-1}\dd\y^A-\frac{\partial^L\HH}{\partial\,\p_A}\right)
	\right\}+\label{proof_useful_relation}\\
	& &+i_{\bxi}\left(\frac{\partial^R\HH}{\partial\,\y^A}\wedge\frac{\partial^L\HH}{\partial\,\p_A}\right)\nonumber
	\een
	The last term is vanishing since the object in the bracket is horizontal $(n+1)$-form over $\M$. The rest is obvious. If the pullback of $i_{\vv}\msf$ along $\sect\in\Sec(\M,\Phase)$ is vanishing for arbitrary $\vv$ then
	\be
	\sect^*(i_{\vv}\msf)=0\;\;\Rightarrow\;\; \sect^*\left((-1)^{n-q-1}\dd\y^A-\frac{\partial^L\HH}{\partial\,\p_A}\right)=0\text{ and }\sect^*\left(-\,\dd\p_A-\frac{\partial^R\HH}{\partial\,\y^A}\right)=0
	\ee
	which means the section $\sect$ must satisfy the Hamilton equations (\ref{Hamilton_equations}). And vice versa if the equations (\ref{Hamilton_equations}) are satisfied then $\sect^*(i_{\vv}\msf)=0$ for arbitrary $\vv\in\Sec(\Phase,\Tan\Phase)$.
\qed\end{proof}
Let  $\GG$ be a certain subgroup of the group of all d-jet diffeomorphisms $\dDiff{\Phase}$ on the phase bundle $\Phase\to\M$ and let for all flows $\bleta_{\Phase}:\GG\times\R\to\GG$ there exists generating vector field $\vv\in\Sec(\Phase,\Tan\Phase)$. Linear span of such vector fields forms Lie algebra $\alg{\GG}\subset\Sec(\Phase,\Tan\Phase)$ of the group $\GG$. The group is called symmetry of the physical system if every element $\bleta_{\Phase}\in\GG$ keeps the dynamical Cartan-Poincar\' e $n$-form $\mef$ invariant
\ben
(\bleta_{\Phase}^{-1})^*\mef=\mef \label{symmetry_invariant_mef}
\een
For arbitrary vector $\vv\in\alg{\GG}$ this condition yields
\ben
\lder{\vv}\mef=0 \label{symmetry_invariant_mef_lie_derivative}.
\een
Let $\Q\in\Sec(\Phase,\forms^{n-1}\Phase)$ be a $(n-1)$-form on $\Phase$ and let there exists a vector field $\vv\in\Sec(\Phase,\Tan\Phase)$ such that 
\be
i_{\vv}\msf=\dd\Q,
\ee
then $\vv$ is the vector field associated to the $(n-1)$-form $\Q$. Since for $\forall\sect\in\Sec(\M,\Phase)$
\ben
\sect^*(i_{\vv}\msf)=\dd\sect^*\Q,\label{pullback_of_d_general_observable}
\een
the theorem \ref{theorem_eom_msf} implies that if $\sect$ is the solution of the Hamilton equations (\ref{Hamilton_equations}) then $\dd\sect^*\Q=0$, i.e. $\sect^*\Q$ is a current of conserved quantity. This yields the famous Noether's theorem.
\begin{theorem}\label{Neother_theorem}
	Let $\GG$ be the symmetry group of the physical system and let $\sect$ be the solution of the Hamilton equations (\ref{Hamilton_equations}) then for arbitrary $\vv\in\alg{\GG}$ there exists current $\JJ_{\vv}=(i_{\vv}\mef)$, called Noether's current, such that the equation $\dd\sect^*\JJ_{\vv}=0$ is satisfied.
\end{theorem}
\begin{proof}
	Condition (\ref{symmetry_invariant_mef_lie_derivative}) implies
	\be
	0=\lder{\vv}\mef=i_{\vv}\dd\mef+\dd i_{\vv}\mef\;\;\;\;\Rightarrow\;\;\;\;-i_{\vv}\dd\mef=i_{\vv}\msf=\dd i_{\vv}\mef
	\ee
	This with the comments above proves the theorem.
\qed\end{proof}
Now, turn our attention to the converse of this theorem. In the case when the symmetry group $\GG$ is sufficiently large, this of course depends on the considered physical system, the answer is in affirmative. Let $\algx{\GG}{\p}=\Span\big\{\vv\vb{\p}\in\Tan_{\p}\Phase;\vv\in\alg{\GG}\big\}$ denote a vector space spanned on vectors of the Lie algebra $\alg{\GG}$ settled in the point $\p\in\Phase$. Group $\GG\subset\dDiff{\Phase}$ is called vertically transitive\cite{Momentum_maps} if $\Ver_{\p}\Phase\subset\algx{\GG}{\p}$ for $\forall\p\in\Phase$.
\begin{theorem}\label{Converse_Noether_theorem_diff_version}
	Let $\GG$ be vertically transitive symmetry group of the system and let for $\sect\in\Sec(\M,\Phase)$ the condition
	\be
	\forall\vv\in\alg{\GG}:\;\dd\sect^*\JJ_{\vv}=0
	\ee
	be satisfied, then $\sect$ is a solution of the Hamilton equations (\ref{Hamilton_equations}).
\end{theorem}
\begin{proof}
	Since $\Ver_{\p}\Phase\subset\algx{\GG}{\p}$ for $\forall\p\in\Phase$ then for all $\vv\in\alg{\GG}$ such that $\vv\vb{\p}=\vv^A\partial_{\y^A}+\ww_A\partial_{\p_A}\in\Ver_{\p}{\Phase}$ one gets
	\be
	\dd\JJ_{\vv}\vb{\p}=i_{\vv}\msf\vb{\p}=\Bigg(\ww_A\wedge\left((-1)^{n-q-1}\dd\y^A-\frac{\partial^L\HH}{\partial\,\p_A}\right)-
	\left(-\,\dd\p_A-\frac{\partial^R\HH}{\partial\,\y^A}\right)\wedge\vv^A\Bigg)\Bigg|_{_{\p}}.
	\ee
	The pullback along the section $\sect\in\Sec(\M,\Phase)$ gives
	\be
	\big(\dd\sect^*\JJ_{\vv})\big|_{_{\xxx}}
	&=&\phantom{-}\sect^*\Bigg(\ww_A\wedge\left((-1)^{n-q-1}\dd\y^A-\frac{\partial^L\HH}{\partial\,\p_A}\right)-
	\left(-\,\dd\p_A-\frac{\partial^R\HH}{\partial\,\y^A}\right)\wedge\vv^A\Bigg)\Bigg|_{_{\xxx}}\\
	&=&\phantom{-}(\sect^*\ww_A)\big|_{_{\xxx}}\wedge\Bigg(\sect^*\left((-1)^{n-q-1}\dd\y^A-\frac{\partial^L\HH}{\partial\,\p_A}\right)\Bigg)\Bigg|_{_{\xxx}}\\
	&&-\Bigg(\sect^*\left(-\,\dd\p_A-\frac{\partial^R\HH}{\partial\,\y^A}\right)\Bigg)\Bigg|_{_{\xxx}}
	\wedge(\sect^*\vv_A)\big|_{_{\xxx}}
	\ee
	Now, it is easy to see, that the arbitrariness of $\vv\vb{\p}\in\Ver_{\p}(\Phase)$ and the condition $\big(\dd\sect^*\JJ_{\vv}\big)\big|_{_{\xxx}}=0$ for all $\xxx\in\M$ imply that $\sect$ is the solution of the Hamilton equations (\ref{Hamilton_equations}).
\qed\end{proof}
Group $\GG\subset\dDiff{\Phase}$ is called localizable if for every pair $(\BSigma_1,\BSigma_2)$ of disjoint embeddings of $\BSigma$ and every 
vector $\vv\in\alg{\GG}$ there exists vector $\ww\in\alg{\GG}$ such that $\vv\vb{\BSigma_1}=\ww\vb{\BSigma_1}$ and $\ww\vb{\BSigma_2}=0$.
Localizable symmetry group is called gauge group. With these in hands we can formulate the second Noether's theorem of integral formulation of the equations of the motion.
\begin{theorem}\label{Converse_Noether_theorem_integral_version}
	Let $\GG$ be the vertically transitive gauge group of the physical system. Section $\sect\in\Sec(\M,\Phase)$ is the solution of the Hamilton equations of motion (\ref{Hamilton_equations}) if and only if for $\forall\vv\in\alg{\GG}$ and all embeddings $\BSigma_0$ of $\BSigma$ into $\M$ the integral
	\be
	\int\limits_{\BSigma_0}\sect^*\JJ_{\vv}=0
	\ee
	is vanishing.
\end{theorem}
\begin{proof}
	If the integral $\int\limits_{\BSigma_0}\sect^*\JJ_{\vv}=0$ is vanishing for all embeddings $\BSigma_0$ and all currents $\JJ_{\vv}$ then we have for  arbitrary foliation $\M_{\Ii}\simeq\BSigma\times\Ii$, where $\Ii=\langle t_{\text{ini}},t_{\text{fin}}\rangle$, an identity
	\be
	0&=&\int\limits_{\BSigma_{\text{fin}}}\sect^*\JJ_{\vv}-\int\limits_{\BSigma_{\text{ini}}}\sect^*\JJ_{\vv}
	=\oint\limits_{\partial\M_{\Ii}}\sect^*\JJ_{\vv}\\
	&=&\int\limits_{\M_{\Ii}}\dd(\sect^*\JJ_{\vv}).
	\ee
	The arbitrariness of the foliation $\M_{\Ii}$ implies for all $\JJ_{\vv}$
	\be
	0=\dd(\sect^*\JJ_{\vv})
	\ee
	and due to theorem \ref{Converse_Noether_theorem_diff_version} the section $\sect$ is the solution of the Hamilton equations of motion.\\
	\hs Conversely if the section $\sect\in\Sec(\M,\Phase)$ is the solution of the Hamilton equations then $0=\dd(\sect^*\JJ_{\vv})$ for all $\JJ_{\vv}$ and therefore
	\be
	0=\int\limits_{\BSigma_{\text{fin}}}\sect^*\JJ_{\vv}-\int\limits_{\BSigma_{\text{ini}}}\sect^*\JJ_{\vv}.
	\ee
	Since $\GG$ is localizable we may choose $\vv$ being vanishing on $\BSigma_{\text{ini}}$ and hence
	\be
	0=\int\limits_{\BSigma_{\text{fin}}}\sect^*\JJ_{\vv}
	\ee
	and arbitrariness of $\BSigma_{\text{fin}}$ yields the result.
\qed\end{proof}
NOTE: The compactness of $\BSigma$ plays crucial role in the theorem, if $\BSigma$ is noncompact then one must take into account boundary terms, e.g. for asymptotically flat spacetimes these yield in general nonvanishing global charges (Energy-momentum "tensor" at infinity, total electric charge,\dots). 

\section{Local Covariant Poisson Brackets}\label{LGPB}
In the canonical quantization Poisson bracket plays a crucial role. Indeed, a certain set of basic kinematical variables equipped with Poisson bracket forms an algebra, called fundamental, and the first step of any approach to the quantization is to find its representation on an appropriate Hilbert space. Therefore it is important to introduce Poisson bracket also in the field theory. As we will see we can define two types of brackets in the field theory $(n\ge 2)$. The first one, called local covariant Poisson bracket, is defined in every point of the phase space. Representation of the local fundamental algebra can be viewed as the first quantization in the context of the covariant quantization, but usually this step is called pre-quantization\cite{Kan3,Kan4}. The task of the section is to define local covariant Poisson bracket among covariant observables and also explore some of their algebraic properties.\\
\hs In the previous section we have introduced two basic multisymplectic manifolds associated with the given physical system. The first one given by the d-jet dual $\JY^*$ and its canonical multisymplectic form $\kmsf$ defined in (\ref{msf_dual}) was playing role in the kinematical description. The second one given by the image of the Legendre transformation $\Phase=\LTr(\JY)$ equipped with the multisymplectic form $\msf$ defined in (\ref{msf_dynamical}) was used for the formulation of the Hamilton equations of motion (\ref{Hamilton_equations}). For a while we are not going to distinguish between them and explore some of their algebraic properties simultaneously. Let $\proj{\M}{\Fibre}:\Fibre\to\M$ be given fibre bundle over the spacetime $\M$, e.g. $\Fibre=\JY^*$ or $\Fibre=\Phase$, and let there be given exact $(n+1)$-multisymplectic form $\kmsf=-\dd\kmef$ over it, e.g. $\kmsf=\kmsf$ for $\JY^*$ or $\kmsf=\msf$  for $\Phase$. In standard symplectic space there is associated the hamiltonian vector field to every hamiltonian observable and it is similar in the multisymplectic case. We say that $\vv\in\Sec(\Fibre,\Tan\Fibre)$ is hamiltonian vector field if there exists hamiltonian $(n-1)$-form, or observable, $\Q\in\Sec(\Fibre,\forms^{n-1}\Fibre)$ such that
\ben
i_{\vv}\kmsf=\dd\Q.\label{Hamiltonian_relation}
\een
Set of all hamiltonian vector fields or set of all hamiltonian $(n-1)$-forms is denoted by $\Ham^1_{\kmsf}\Tan\Fibre$ or $\Ham_{\kmsf}^{n-1}\Fibre$, respecti\-velly. $(\Q,\vv)$ is called hamiltonian pair on $\Fibre$ if $(\Q,\vv)\in\Ham_{\kmsf}^{n-1}\Fibre\times\Ham^1_{\kmsf}\Tan\Fibre$ and $(\Q,\vv)$ are related by (\ref{Hamiltonian_relation}).\\
\hs The local covariant Poisson bracket for observables $\Q_1,\Q_2\in\Ham_{\kmsf}^{n-1}\Fibre$ is defined by \cite{FPR}
\ben
\lpb\Q_1,\Q_2\rpb=i_{\vv_1} i_{\vv_2}\blomega+\dd\left(i_{\vv_1}\Q_2-i_{\vv_2}\Q_1-i_{\vv_1} i_{\vv_2}\mef\right),\label{GPB_n_1}
\een
where $\vv_1,\vv_2\in\Ham_{\kmsf}^1\Tan\Fibre$ are hamiltonian vectors related to $\Q_1,\Q_2$, respectively. The local covariant Poisson bracket defined in \cite{FPR} is equal up to the sign to local covariant Poisson bracket defined here and bracket in \cite{FPR} is also defined for, so-called, Poisson forms with degrees less then $n-1$, but it is sufficient for our purposes to consider only (\ref{GPB_n_1}). Local Poisson bracket defined by (\ref{GPB_n_1}) satisfied Jacobi identity, \cite{FPR} theorem 3.8, 
\be
\lpb\Q_1,\lpb\Q_2,\Q_3\rpb\rpb+\lpb\Q_3,\lpb\Q_1,\Q_2\rpb\rpb+\lpb\Q_2,\lpb\Q_3,\Q_1\rpb\rpb=0,
\ee
where $\Q_1,\Q_2,\Q_3\in\Ham_{\kmsf}^{n-1}\Fibre$. We should note here that $(\Ham^{n-1}_{\kmsf}\Fibre,\lpb.,.\rpb)$ is Lie algebra only, since there is no multiplication on $\Ham_{\kmsf}^{n-1}\Fibre$.\\
\hs We have seen in the previous section that Noether's currents $\JJ_{\vv}$ are generated by vector fields preserving Cartan-Poincar\' e form $\kmef$, i.e. if $\lder{\vv}\kmef=0$ then for $i_{\vv}\kmef=\JJ_{\vv}$ we have $i_{\vv}\kmsf=\dd\JJ_{\vv}$. There already exists Lie algebra defined on 
the set of all Noether's vector fields with product given by Lie bracket $\lsn,\rsn$ on $\Tan\Fibre$. This algebra is naturally represented on the set of all Noether's currents.
\begin{lemma}\label{lemma_Noether_charges_and_their_LPB}
	Let $\vv_1,\vv_2$ be Noether's vectors, i.e. satisfying $\lder{\vv_1}\kmef=\lder{\vv_2}\kmef=0$, then (\ref{GPB_n_1}) yields
	\ben
	\lpb\JJ_{\vv_1},\JJ_{\vv_2}\rpb=\JJ_{\lsn\vv_1,\vv_2\rsn},\label{local_Poisson_bracket_Noethers_charges_general_form}
	\een
	where $\lsn .,.\rsn$ means Lie bracket on $\Tan\Fibre$.
\end{lemma}
\begin{proof}
	Direct calculation yields
	\be
	\lpb\JJ_{\vv_1},\JJ_{\vv_2}\rpb&=&-i_{\vv_1}i_{\vv_2}\dd\kmef-\dd i_{\vv_2}i_{\vv_1}\kmef\\
	&=&-i_{\vv_2}\dd i_{\vv_1}\kmef-\dd i_{\vv_2}i_{\vv_1}\kmef\\
	&=&-\lder{\vv_2}i_{\vv_1}\kmef=i_{\lsn\vv_1,\vv_2\rsn}\kmef=\JJ_{\lsn\vv_1,\vv_2\rsn},
	\ee
	where we used relation $i_{\lsn\vv_1,\vv_2\rsn}=\lder{\vv_1}i_{\vv_2}-i_{\vv_2}\lder{\vv_1}$.
\qed\end{proof}
We have constructed two different multisymplectic structures for given physical system $(\Y,\LL)$. Now, we are curious how these structures are related. Let $(\Fibre=\JY^*, \kmsf)$ with related  algebra of observables $(\Ham_{\kmsf}^{n-1}\Fibre,\lpb\,,\,\rpb_{\kmsf})$ 
or $(\Phase,\msf)$ with $(\Ham_{\msf}^{n-1}\Phase,\lpb\,,\,\rpb_{\msf})$ be kinematical or dynamical multisymplectic manifold, respectively. Since $\Fibre$ or $\Phase$ can be locally coordanized by $(x^{\mu},\y^A,\p_A,\h)$ or $(x^{\mu},\y^A,\p_A)$, respectively, we have a bundle chain
\be
\Fibre\stackrel{\proj{\Phase}{\Fibre}}{\longrightarrow}\Phase\stackrel{\proj{\M}{\Phase}}{\longrightarrow}\M\text{\hs and \hs}\proj{\M}{\Fibre}=\proj{\M}{\Phase}\circ\proj{\Phase}{\Fibre}.
\ee
Phase space $\Phase$ can be embedded into the kinematical multisymplectic space $\Fibre$ by Hamilton map $\Hmap:\Phase\to\Fibre$ defined by
\ben
\Hmap:(\xxx,\y^A,\p_A)\mapsto(\xxx,\y^A,\p_A,\h=\HH),\label{Hamilton_map_definition}
\een
where $\HH$ is Hamiltonian defined by (\ref{general_hamiltonian}).\\
\hs Let $\vv\in\Tan\Fibre$ or $\vv\in\Tan\Phase$ then we say that $\vv$ is $\M$-vertical if $(\proj{\M}{\Fibre})_*\vv=0$ or $(\proj{\M}{\Phase})_*\vv=0$ and we denote subbundle of such vectors by $\MVer\Fibre\subset\Tan\Fibre$ or $\MVer\Phase\subset\Tan\Phase$, respectively. Vector $\vv\in\Tan\Fibre$ is called $\HH$-vertical if $(\proj{\Phase}{\Fibre})_*\vv=0$ and subbundle of such vectors is denoted by $\HVer\Fibre\subset\Tan\Fibre$. $(\proh)_*\vv$ means pushforward of the vector $\vv$ along map $\proh$. Let $\alpha\in\forms\Fibre$ or $\alpha\in\forms\Phase$ we say that $\alpha$ is $\M$-horizontal if $\alpha$ is annihilated by all $\M$-vertical vectors and subbundle of such forms is denoted by $\Mhor\Fibre\subset\forms\Fibre$ or $\Mhor\Phase\subset\forms\Phase$, respectivelly. Let $\alpha\in\forms\Fibre$ we say that $\alpha$ is $\HH$-horizontal if $\alpha$ is annihilated by all $\HH$-vertical vectors and subbundle of such forms is denoted by $\Hhor\Fibre\subset\forms\Fibre$.\\
\hs Let $\Q\in\Ham^{n-1}_{\kmsf}\Fibre$ be kinematical observable on $\Fibre$ and $\vv\in\Ham^1_{\kmsf}\Tan\Fibre$ be its related hamiltonian vector. Since relation between bundles $\Fibre$ and $\Phase$ is given by projection $\proj{\Phase}{\Fibre}$ and Hamilton map $\Hmap$ it is important to explore how general $\Q$ depends on $\h$. Let us denote $\h=\bar{h}\dd\Sigma$. We can decompose $\Q$ as 
\ben
\Q=\alpha+\dd\bar{h}\wedge\beta,\label{h-decomposition_of_kinematical_observable}
\een
where $\alpha\in\Sec(\Fibre,\Hhor^{n-1}\Fibre)$ and $\beta\in\Sec(\Fibre,\Hhor^{n-2}\Fibre)$ are certain $\HH$-horizontal form fields on $\Fibre$. Following lemma shows that we can consider each obserable $\Q\in\Ham_{\kmsf}^{n-1}\Fibre$ to be $\HH$-horizontal without lack of generality and therefore $\Ham^{n-1}_{\kmsf}\Fibre$ denotes a set of such $\HH$-horizontal observables when the proof of the statement is done. 
\begin{lemma}\label{H-horizontality_of_observable}
	Let $\Q\in\Ham^{n-1}_{\kmsf}\Fibre$ be general kinematical observable on $(\Fibre,\kmsf)$ and $\vv=\bxi+\ww+\uu\partial_{\h}\in\Ham_{\kmsf}^1\Tan\Fibre$ be its hamiltonian vector field then upto exact term
	\ben
	\Q=\Q_{\Phase}+i_{\bxi}(\h-\HH),\label{formula_H-horizontality_of_observable}
	\een
	where $\Q_{\Phase}$ is $\h$-independent $\HH$-horizontal $(n-1)$-form.
\end{lemma}
\begin{proof}Let $\Q\in\Ham^{n-1}_{\kmsf}\Fibre$ and let $\vv\in\Ham^1_{\kmsf}\Tan\Fibre$ be its associated hamiltonian vector field, which can be written as
	\ben
	\vv&=&\xi^{\mu}\partial_{\mu}+w^A_{(\mu)_q}\partial_{y^A_{(\mu)_q}}+w_{A(\mu)^{q+1}}\partial_{p_{A{(\mu)^{q+1}}}}+\bar{u}\partial_{\bar{h}},\nonumber\\
	&=&\bxi+\ww^A\partial_{\y^A}+\ww_A\partial_{\p_A}+\uu\partial_{\h},\label{auxiliar_equation_iii}\\
	&=&\bxi+\ww+\uu\partial_{\h}.\nonumber
	\een
	(\ref{Hamiltonian_relation}) and (\ref{h-decomposition_of_kinematical_observable}) yield
	\ben
	\label{auxiliar_equation_iv}
	\begin{tabular}{lcl}
		$\dd^{\Phase}\alpha+\dd\bar{h}\wedge\big(\frac{\partial\alpha}{\partial\bar{h}}-\dd^{\Phase}\beta\big)$
		&$=$&$\phantom{+}
		i_{\bxi}(-1)^{n-q-1}\dd\p_A\wedge\dd\y^A+\dd\bar{h}\wedge i_{\bxi}\dd\Sigma$\\
		&&$+i_{\ww}(-1)^{n-q-1}\dd\p_A\wedge\dd\y^A-\uu,$
	\end{tabular}
	\een
	where we used notation \[\dd^{\Phase}\alpha=\dd x^{\mu}\wedge\lder{\partial_{\mu}}\alpha+\dd y^A_{(\mu)_q}\wedge
	\lder{\partial_{y^A_{(\mu)_q}}}\alpha+
	\dd p_{A(\mu)^{q+1}}\wedge\lder{\partial_{p_{A(\mu)^{q+1}}}}\alpha\]
	for short. If we compare terms containing $\dd\bar{h}$ in (\ref{auxiliar_equation_iv}) we get 
	\be
	i_{\bxi}\dd\Sigma=\frac{\partial\alpha}{\partial\bar{h}}-\dd^{\Phase}\beta,
	\ee
	or
	\be
	\Q=\Q_{\Phase}+\int\limits_{\bar{H}}^{\bar{h}}\left(i_{\bxi}\dd\Sigma+\dd^{\Phase}\beta\right)\dd\bar{h}+\dd\bar{h}\wedge\beta,
	\ee
	where $\Q_{\Phase}=(\proj{\Phase}{\Fibre})^*(\Hmap)^*\alpha$ is $\bar{h}$-independent $\HH$-horizontal $(n-1)$-form and $\HH=\bar{H}\dd\BSigma$. If we set $\gamma=\int\limits_{\bar{H}}^{\bar{h}}\dd\bar{h}\,\beta$ then
	\be
	\Q=\Q_{\Phase}+\int\limits_{\bar{H}}^{\bar{h}}\left(i_{\bxi}\dd\Sigma\right)\dd\bar{h}+\dd^{\Phase}\gamma+
	\dd\bar{h}\wedge\frac{\partial\gamma}{\partial\bar{h}}=\Q_{\Phase}+\int\limits_{\bar{H}}^{\bar{h}}\left(i_{\bxi}\dd\Sigma\right)\dd\bar{h}+\dd\gamma.
	\ee
	Its exterior derivative is
	\be
	\dd\Q=\dd^{\Phase}\Q_{\Phase}+\dd\bar{h}\wedge i_{\bxi}\dd\Sigma+\dd^{\Phase}\int\limits_{\bar{H}}^{\bar{h}}\left(i_{\bxi}\dd\Sigma\right)\dd\bar{h}.
	\ee
	Since the last term is at most linear in $\dd\y^A$ and $\dd\p_A$ we see that the term $\dd^{\Phase}\Q_{\Phase}$ contains complete information about term $i_{\bxi}(-1)^{n-q-1}\dd\p_A\wedge\dd\y^A$ in (\ref{auxiliar_equation_iv}), hence $\bxi$ does not depent on $\bar{h}$. Thus, we can write
	\be
	\Q=\Q_{\Phase}+i_{\bxi}(\h-\HH)+\dd\gamma.
	\ee
\qed\end{proof}
The next important property of hamiltonian flows on $\Fibre$ or $\Phase$ is that they define homomorphisms on the bundle $\Fibre\to\Phase\to\M$ or $\Phase\to\M$, respectively. 
\begin{lemma}
	Let $\vv\in\Ham^1_{\kmsf}\Tan\Fibre$ be hamiltonian vector of observable $\Q^{\Fibre}\in\Ham^{n-1}_{\kmsf}\Fibre$ then the chain of vectors $\vv=\bxi+\ww+\uu\partial_{\h}\to\bxi+\ww\to\bxi$ is projectable on $\Fibre\to\Phase\to\M$.\\
	\hs Let $\VV\in\Ham^1_{\msf}\Tan\Phase$ be hamiltonian vector of observable $\Q^{\Phase}\in\Ham^{n-1}_{\msf}\Phase$ then the chain of vectors $\VV=\Bxi+\WW\to\Bxi$ is projectable on $\Phase\to\M$.
\end{lemma}
\begin{proof}Relation (\ref{Hamiltonian_relation}) yields $\dd i_{\vv}\kmsf=0$ for $\Fibre$ or $\dd i_{\VV}\msf=0$ for $\Phase$, respectively. We have
	\ben
	\dd\left[i_{\bxi}(\!-\!1)^{n-q-1}\dd\p_A\wedge\dd\y^A\right]\!-\!\dd i_{\bxi}\dd\h\!+\!
	(\!-\!1)^{n-q-1}\left(\dd\ww_A\wedge\dd\y^A\!+\!\dd\p_A\wedge\dd\ww^A\right)\!-\!\dd\uu=0\label{auxiliar_equation_v}
	\een
	for $\Fibre$ or
	\ben
	\dd\left[i_{\Bxi}(\!-\!1)^{n-q-1}\dd\p_A\wedge\dd\y^A\right]\!-\!\dd i_{\VV}\dd\HH\!+\!
	(\!-\!1)^{n-q-1}\left(\dd\WW_A\wedge\dd\y^A\!+\!\dd\p_A\wedge\dd\WW^A\right)=0\label{auxiliar_equation_vi}
	\een
	for $\Phase$, respectively.\\
	\hs Since $\bxi$ does not depend on $\h$ and only the first term in (\ref{auxiliar_equation_v}) is cubic in $\dd y^A_{(\mu)_q}$ and $\dd p_{A\,(\mu)^{q+1}}$ we have
	\be
	(-1)^{n-q-1}\left(\dd y^B_{(\nu)_q}
	\frac{\partial\xi^{\nu}}{\partial y^B_{(\nu)_q}}+
	\dd p_{B\,(\nu)^{q+1}}
	\frac{\partial\xi^{\nu}}{\partial p_{B\,(\nu)^{q+1}}}
	\right)\wedge i_{\partial_{\nu}}\left(\dd\p_A\wedge\dd\y^A\right)=0.
	\ee
	We see that $\bxi$ depends on $x^{\mu}$ only. The second term in (\ref{auxiliar_equation_v}) can be canceled only by $\dd\bar{h}$-part of the forth term.
	This yields that $\dd\bar{h}$-part of the third term should be vanishing, i.e. $\ww$ does not depend on $\h$. Therefore the chain of vector fields $\vv\to\bxi+\ww\to\bxi$ on $\Fibre\to\Phase\to\M$ is projectable.\\
	\hs Similarly for $\Phase$ we have
	\be
	(-1)^{n-q-1}\left(\dd y^B_{(\nu)_q}
	\frac{\partial\Xi^{\nu}}{\partial y^B_{(\nu)_q}}+
	\dd p_{B\,(\nu)^{q+1}}
	\frac{\partial\Xi^{\nu}}{\partial p_{B\,(\nu)^{q+1}}}
	\right)\wedge i_{\partial_{\nu}}\left(\dd\p_A\wedge\dd\y^A\right)=0
	\ee
	and therefore the chain $\VV=\Bxi+\WW\to\Bxi$ is projectable on $\Phase\to\M$.
\qed\end{proof}
\begin{lemma}
	Let $\Q_1^{\Fibre},\Q_2^{\Fibre}\in\Ham^{n-1}_{\kmsf}\Fibre$ and $(\proj{\Phase}{\Fibre})^*(\Hmap)^*\Q_1^{\Fibre}=(\proj{\Phase}{\Fibre})^*(\Hmap)^*\Q_2^{\Fibre}$ then $\Q_1^{\Fibre}=\Q_2^{\Fibre}$.
\end{lemma}
\begin{proof}
	Let $\vv_i=\bxi_i+\ww_i+\uu_i\partial_{\h}$ be the decomposition, similar to (\ref{auxiliar_equation_iii}), of the hamiltonian vector $\vv_i\in\Ham^1_{\kmsf}\Tan\Fibre$ associated to the observable $\Q_i^{\Fibre}\in\Ham^{n-1}_{\kmsf}\Fibre$, where $i=1,2$. Each $\Q_i^{\Fibre}$ can be written as
	\be
	\Q^{\Fibre}_i=\Q_{\Phase}+\Q_i\text{, where\hs} \Q_{\Phase}=(\proj{\Phase}{\Fibre})^*(\Hmap)^*\Q_i^{\Fibre}.
	\ee
	For $\Q_i=i_{\bxi_i}(\h-\HH)$ we get
	\ben
	\dd\Q^{\Fibre}_i=\dd\Q_{\Phase}+\dd i_{\bxi_i}(\h-\HH)\label{auxiliar_IV}
	\een
	and also
	\ben
	i_{\vv_i}\kmsf=-i_{\bxi_i}\dd\h+i_{\bxi_i}(-1)^{n-q-1}\dd\p_A\wedge\dd\y^A+i_{\ww_i}(-1)^{n-q-1}\dd\p_A\wedge\dd\y^A-\uu_i.\label{auxiliar_V}
	\een
	Since the second term $i_{\bxi_i}(-1)^{n-q-1}\dd\p_A\wedge\dd\y^A$ in (\ref{auxiliar_V}) can be canceled only by part of $\dd\Q_{\Phase}$ in (\ref{auxiliar_IV}) we have that $\bxi_1=\bxi_2$ and therefore $\Q_1=\Q_2$.
\qed\end{proof}
This lemma shows that there exists at most one observable $\Q^{\Fibre}\in\Ham_{\kmsf}^{n-1}\Fibre$ for given form $\Q^{\Phase}\in\Sec(\Phase,\forms^{n-1}\Phase)$ over $\Phase$ and vice versa $\Ham^{n-1}_{\kmsf}\Fibre$ defines subspace of kinematical observables $\Ham_{\kmsf}^{n-1}\Phase=(\Hmap)^*\Ham^{n-1}_{\kmsf}\Fibre$ in the set of all $(n-1)$-forms $\Sec(\Phase,\forms^{n-1}\Phase)$ on dynamical phase space $\Phase$. Let $\Q^{\Phase}_i\in\Ham_{\kmsf}^{n-1}\Phase$ and let $\Q^{\Fibre}_i$ be its kinematical extension on $\Fibre$ then we can define the local kinematical Poisson bracket on $\Ham_{\kmsf}^{n-1}\Phase$ by
\ben
\lpb\Q^{\Phase}_1,\Q^{\Phase}_2\rpb^{\Phase}_{\kmsf}=(\Hmap)^*\lpb\Q^{\Fibre}_1,\Q^{\Fibre}_2\rpb_{\kmsf}\label{Phase_space_kinematical_PB}.
\een
and get Lie algebra of kinematical observables $(\Ham^{n-1}_{\kmsf}\Phase,\lpb.,.\rpb^{\Phase}_{\kmsf})$ on the dynamical phase space $\Phase$.\\
\hs It seems that we have two local Poisson brackets on dynamical phase space $\Phase$: kinematical algebra $(\Ham^{n-1}_{\kmsf}\Phase,\lpb.,.\rpb_{\kmsf}^{\Phase})$ just defined and $(\Ham^{n-1}_{\msf}\Phase,\lpb.,.\rpb_{\msf})$ as dynamical algebra. 
Problem is that the dynamical bracket $\lpb.,.\rpb^{\Phase}_{\msf}$ depends on the Hamiltonian. Due to (\ref{pullback_of_d_general_observable}) and
theorem \ref{theorem_eom_msf} the dynamical observables are constants of motion we are not able to measure any local information deppending on time evolution only by them. We will show in the next lemma \ref{Lemma_on_dynamical_subalgebra} that the dynamical algebra $(\Ham^{n-1}_{\msf}\Phase,\lpb.,.\rpb_{\msf})$ is subalgebra in $(\Ham^{n-1}_{\kmsf}\Phase,\lpb.,.\rpb^{\Phase}_{\kmsf})$ therefore the bracket defined by (\ref{Phase_space_kinematical_PB}) extends the algebra of observables $(\Ham^{n-1}_{\msf}\Phase,\lpb.,.\rpb_{\msf})$ on the dynamical phase space $\Phase$.\\
\hs Before we do that let us introduce following notation. If we denote
\be
\dd\bxi=\xi^{\mu}_{,\nu}\dd x^{\nu}\otimes\partial_{\mu}
\ee
for $\xi\in\Sec(\M,\Tan\M)$ then we can consider $\dd\bxi$ to be object, not tensor, of the set of vector valued forms\cite{Theory_of_vector_valued_differential_forms}, where the interior product of vector valued $r$-form $\beta=\sum\limits_a\alpha_a\otimes\vv_a\in\forms^r_{\xxx}\X\otimes\Tan_{\xxx}\X$ over $\xxx\in\X$, where $\X$ is general manifold, and $\forall a:\alpha_a\in\forms^r_{\xxx}\X,\vv_a\in\Tan_{\xxx}\X$ is definend for $\forall\theta\in\forms_{\xxx}\X$ by 
\be
i_{\beta}\theta=\sum\limits_a\alpha_a\wedge i_{\vv_a}\theta.
\ee
\begin{lemma}\label{Lemma_on_dynamical_subalgebra}
	$(\Ham^{n-1}_{\msf}\Phase,\lpb.,.\rpb_{\msf})$ is subalgebra in $(\Ham^{n-1}_{\kmsf}\Phase,\lpb.,.\rpb^{\Phase}_{\kmsf})$.
\end{lemma}
\begin{proof}
	Let $\Q^{\Phase}\in\Ham^{n-1}_{\msf}\Phase$ and $\VV\in\Ham^1\Tan\Phase$ be its associated hamiltonian vector. Goal is to find $\Q^{\Fibre}\in\Ham^{n-1}_{\kmsf}\Fibre$ with associated vector $\vv=\bxi+\ww+\uu\partial_{\h}\in\Ham^1_{\kmsf}\Fibre$ such that $(\Hmap)^*\Q^{\Fibre}=\Q^{\Phase}$. We already know that if $\Q^{\Fibre}$ exists then it looks like $\Q^{\Fibre}=(\proj{\Phase}{\Fibre})^*\Q^{\Phase}+i_{\bxi}(\h-\HH)$. We have
	\be
	i_{\vv}\kmsf&=&\dd\Q^{\Fibre},\\
	i_{\vv}\big((\proj{\Phase}{\Fibre})^*\msf-\dd(\h-\HH)\big)&=&(\proj{\Phase}{\Fibre})^*(i_{\VV}\msf)+i_{\dd\bxi}(\h-\HH)-i_{\bxi}\dd(\h-\HH),\\
	i_{\vv}\big((\proj{\Phase}{\Fibre})^*\msf\big)-\uu+i_{\ww}\dd\HH&=&(\proj{\Phase}{\Fibre})^*(i_{\VV}\msf)+i_{\dd\bxi}(\h-\HH).
	\ee
	Since $\vv$ is $\HH$-projectable the last equation yields $(\proj{\Phase}{\Fibre})_*\vv=\VV$ and $\uu=i_{\ww}\dd\HH-i_{\dd\bxi}(\h-\HH)$. Therefore $\Q^{\Fibre}=(\proj{\Phase}{\Fibre})^*\Q^{\Phase}+i_{\bxi}(\h-\HH)$ is wanted extension for all $\forall\Q^{\Phase}\in\Ham_{\msf}^{n-1}\Phase$. We have 
	\be
	\lpb \Q^{\Fibre}_1,\Q^{\Fibre}_2\rpb_{\kmsf}=(\proj{\Phase}{\Fibre})^*\lpb \Q^{\Phase}_1,\Q^{\Phase}_2\rpb_{\msf}+i_{\lsn\bxi_1,\bxi_2\rsn}(\h-\HH)
	\ee
	for $\forall\Q_i^{\Phase}\in\Ham_{\msf}^{n-1}\Phase$ hence definition (\ref{Phase_space_kinematical_PB}) yields directly 
	$\lpb \Q^{\Phase}_1,\Q^{\Phase}_2\rpb_{\kmsf}^{\Phase}=\lpb \Q^{\Phase}_1,\Q^{\Phase}_2\rpb_{\msf}$.
\qed\end{proof}
\hspace{-7pt}We set $(\Ham^{n-1}\Fibre,\lpb.,.\rpb)=(\Ham^{n-1}_{\kmsf}\Fibre,\lpb.,.\rpb_{\kmsf})$ and $(\Ham^{n-1}\Phase,\lpb.,.\rpb)=(\Ham^{n-1}_{\kmsf}\Phase,\lpb.,.\rpb^{\Phase}_{\kmsf})$. Let us construct some observables.
\begin{lemma}\label{Lemma_on_general_smeared_canonical_observable}
	For $\forall\bpi_A\in\Sec(\M,\forms^{n-q-1}\M)$, $\forall\Bxi^A\in\Sec(\M,\forms^{q}\M)$ forms\\
	\begin{tabular}{l p{10.5cm}}
		$\phantom{(}i)$&$\y(\bpi)=\bpi_A\wedge\y^A$,\\
		$\phantom{(}ii)$&$\p(\Bxi)=\p_A\wedge\Bxi^A$,\\
	\end{tabular}\\
	where we write $\bpi_A=\proj{\M}{\Fibre}^{\bstar}\bpi_A$ and similar for $\Bxi^A$,  are observables on $\Fibre$. The one only non-trivial local Poisson bracket is
	\be
	\lpb\bpi_A\wedge\y^A,\p_B\wedge\Bxi^B\rpb=\bpi_A\wedge\Bxi^A.
	\ee
\end{lemma}
\begin{proof}Direct calculation  shows that vectors
	\[\vv(\bpi)=\bpi_A\partial_{\p_A}-(\dd\bpi_A\wedge\y^A)\partial_{\h}\]
	and
	\[\vv(\Bxi)=-\bxi^A\partial_{\y^A}-(-1)^{n-q-1}(\p_A\wedge\Bxi^A)\partial_{\h}\]
	solve equation (\ref{Hamiltonian_relation}) for appropriate observables $\y(\bpi)$ and $\p(\Bxi)$, respectively. Since both observables are $\M$-horizontal and their hamiltonian vectors are vertical on $\Fibre\to\M$ we have from (\ref{GPB_n_1}) immediately
	\[\lpb\p(\Bxi),\p(\Bxi ')\rpb=\lpb\y(\bpi),\y(\bpi ')\rpb=0\]
	and
	\[\lpb\y(\bpi),\p(\Bxi)\rpb=i_{\vv(\bpi)}\dd(\p_A\wedge\Bxi^A)=\bpi_A\wedge\Bxi^A.\]
\qed\end{proof}
\begin{lemma}\label{Noethers_current_diffeomorphism_general_case}
	There exist representation of $\DiffM$ on $\Ham^{n-1}\Fibre$ generated by Noether's currents $\JJ^{\Fibre}(\bxi)$. $\JJ^{\Fibre}(\bxi)$ are defined for every $\bxi\in\alg\DiffM\equiv\Sec(\M,\Tan\M)$ by
	\be
	\JJ^{\Fibre}(\bxi)=(-1)^{n-q}i_{\bxi}\p_A\wedge\dd\y^A+\p_A\wedge\dd i_{\bxi}\y^A+i_{\bxi}\h
	\ee
	and satisfy identities
	\ben
	\Lpb\JJ^{\Fibre}(\bxi),\JJ^{\Fibre}(\bxi ')\Rpb&=&\JJ^{\Fibre}\big(\lsn\bxi,\bxi '\rsn\big),\label{l1}\\
	\Lpb\y(\bpi),\JJ^{\Fibre}(\bxi)\Rpb&=&\y(-\lder{\bxi}\bpi),\label{l2}\\
	\Lpb\p(\Bxi),\JJ^{\Fibre}(\bxi)\Rpb&=&\p(-\lder{\bxi}\Bxi).\label{l3}
	\een
\end{lemma}
\begin{proof}
	Let $\bleta\in\DiffM$ then the map defined on kinematical bundle $\Fibre$ by
	\be
	\bleta_{\Fibre}:(\xxx,\y^A,\p_A,\h)\mapsto
	\Big(\bleta(\xxx),\big(\bleta^{-1}(\xxx)\big)^*\y^A,\big(\bleta^{-1}(\xxx)\big)^*\p_A,\big(\bleta^{-1}(\xxx)\big)^*\h\Big)
	\ee
	can be extended into the d-jet diffeomorphism on $\Fibre$ due to theorem \ref{Theorem_d_jet_homomorphism_extension}. Let $\bleta(\lambda)\in\DiffM$ be one-parameter group generated by vector field $\bxi\in\alg{\DiffM}$ then its d-jet extension $\bleta_{\Fibre}$ is generated by vector
	\be
	\ww(\bxi)=\bxi-i_{\dd\bxi}\y^A\partial_{\y^A}-i_{\dd\bxi}\p_A\partial_{\p_A}-i_{\dd\bxi}\h\partial_{\h}.
	\ee
	Since $\lder{\ww(\bxi)}\kmef=0$ hence $\JJ^{\Fibre}(\bxi)=i_{\ww(\bxi)}\kmef$ is kinematical Noether's charge. We have $\lsn\ww(\bxi),\ww(\bxi ')\rsn=\ww\big(\lsn\bxi,\bxi '\rsn\big)$ therefore lemma \ref{lemma_Noether_charges_and_their_LPB} yields (\ref{l1}). If we use that $\JJ^{\Fibre}(\bxi)$ is Noether's charge then the bracket (\ref{GPB_n_1}) can be reduced into the form
	\be
	\lpb\Q,\JJ^{\Fibre}(\bxi)\rpb=-\lder{\ww(\bxi)}\Q
	\ee
	and relations $\lder{\ww(\bxi)}\y^A=\lder{\ww(\bxi)}\p_A=0$ imply (\ref{l2}) and (\ref{l3}) immediately.
\qed\end{proof}
\begin{theorem}\label{Observable_canonical_decomposition_theorem}
	General observable $\f$ on $\Fibre$ can be locally decomposed as
	\ben
	\f=\JJ^{\Fibre}(\bxi)+\Q+\dd\alpha,\label{Observable_canonical_decomposition_formula}
	\een
	where $\JJ^{\Fibre}(\bxi)$ is given by previous lemma \ref{Noethers_current_diffeomorphism_general_case}, $\Q$ is $\M$-horizontal simply differentiable $(n-1)$-form deppending only on $\xxx,\y^A,\p_A$ and $\alpha$ is arbitrary $(n-2)$-form on $\Fibre$. If each fibre $\Fibre_{\xxx}=\proj{\M}{\Fibre}^{-1}(\xxx)$ is contractible then (\ref{Observable_canonical_decomposition_formula}) holds globally.
\end{theorem}
\begin{proof}
	Let $\uu$ be general hamiltonian vector field associated to observable $\f$ and let $\uu$ be given by $\uu=\bxi+\uu^A\partial_{\y^A}+\uu_A\partial_{\p_A}+\bmu\partial_{\h}$.
	If we define vertical hamiltonian vector field
	\be
	\vv=\uu-\ww(\bxi)=\vv^A\partial_{\y^A}+\vv_A\partial_{\p_A}+\bnu\partial_{\h},
	\ee
	then we can decompose $\f=\JJ^{\Fibre}(\bxi)+\tilde\Q$, where $\tilde\Q$ is independent of $\h$ and it is associated with the vertical hamiltonian vector field $\vv$, i.e. we have
	\ben
	\dd\tilde\Q=i_{\vv}\kmsf=(-1)^{n-q-1}\vv_A\wedge\dd\y^A-\dd\p_A\wedge\vv^A-\bnu.\label{aux_eq_theorem}
	\een
	If we employ notation $z^I=\big(y^A_{(\mu)_q}, p_{A(\mu)^{q+1}}\big)$ then we can write
	\be
	\tilde{\Q}=\q+\dd z^{I_1}\wedge\q_{I_1}+\dots+\dd z^{(I)_{n-1}}\wedge\q_{(I)_{n-1}},
	\ee 
	where $\q,\dots\q_{(I)_{n-1}}$ are $M$-horizontal $(n-1)$-,$\dots$,$0$- forms  on $\Fibre$, respectively. We have
	\be
	\dd\tilde{\Q}&=&\phantom{+}\dd\q-\dd z^I\wedge(\dd z^J\wedge\q_{I,J}+\dd^{\M}\q_I)+\dots\\
	&&+(-1)^{n-1}\dd z^{(I)_{n-1}}\wedge(\dd z^{I_n}\wedge\q_{(I)_{n-1},I_n}+\dd^{\M}\q_{(I)_{n-1}})
	\ee
	where $\dd^{\M}\q=\dd x^{\mu}\wedge\lder{\partial_{\mu}}\q$ is horizontal part of the exterior derivative operator $\dd$. Relation (\ref{aux_eq_theorem}) yields $\dd\tilde{\Q}(\vv_1,\dots,\vv_j)=0$, where $1<j\leq n-1$,  for all vertical vectors $\forall\vv_1,\dots,\vv_n\in\Ver\Fibre$. This is possible only if there exist on a fibre $\Fibre_{\xxx}=\proj{\M}{\Fibre}^{-1}(\xxx)$ locally, or if $\Fibre_{\xxx}$ is contractible then globally, $\M$-horizontal $\tilde{\q}_{(I)_{n-2}}$ such that $\q_{(I)_{n-1}}=(-1)^{n-2}(n-1)\tilde{\q}_{[(I)_{n-2},I_{n-1}]}$. Thus, we can write
	\be
	\tilde{\Q}=\Q_{n-1}&=&\phantom{+}\q+\dd z^{I_1}\wedge\q_{I_1}+\dots+\dd z^{(I)_{n-2}}\wedge\q_{(I)_{n-2}}-(-1)^{n-2}\dd z^{(I)_{n-2}}\wedge\dd^{\M}\tilde{\q}_{(I)_{n-2}}\\
	&&+\dd\big(\dd z^{(I)_{i-2}}\wedge\tilde{\q}_{(I)_{n-2}}\big)\\
	&=&\phantom{+}\Q_{n-2}+\dd\big(\dd z^{(I)_{i-2}}\wedge\tilde{\q}_{(I)_{n-2}}\big).
	\ee
	If we continue in calculation of $\Q_i$'s until $\Q_0$ then we get $\tilde{\Q}=\Q_0+\dd\alpha$, where $\Q=\Q_0$ is $\M$-horizontal $(n-1)$-form, keeping
	\be
	i_{\vv}\kmsf=(-1)^{n-q-1}\vv_A\wedge\dd\y^A-\dd\p_a\wedge\vv^A-\bnu=\dd\Q.
	\ee
	Now, recall the proof of theorem \ref{Theorem_d_jet_homomorphism_extension} we see that $\Q$ is simply differentiable and we also get
	\ben
	\vv_A&=&\phantom{-}\frac{\partial^R\Q}{\partial\y^A},\nonumber\\
	\vv^A&=&-\frac{\partial^L\Q}{\partial\p_A},\label{expression_of_hamiltonian_vector_for_horizontal_observable}\\
	\bnu^{\phantom{A}}&=&-\dd^{\M}\Q.\nonumber
	\een
\qed\end{proof}
The local Poisson brackets among two $\M$-horizontal observables $\Q_1,\Q_2$ can be expressed by relations (\ref{expression_of_hamiltonian_vector_for_horizontal_observable}) via
\be
\lpb\Q_1,\Q_2\rpb=\phantom{-}\frac{\partial^R\Q_1}{\partial\y^A}\wedge\frac{\partial^L\Q_2}{\partial\p_A}
-\frac{\partial^R\Q_2}{\partial\y^A}\wedge\frac{\partial^L\Q_1}{\partial\p_A}.
\ee
This expression mimics the well known formula of the classical mechanics.\\
\hs Theorem \ref{Observable_canonical_decomposition_theorem} shows that general observable $\f$ on $\Fibre$ can be locally decomposed up to the exact term as $\f=\JJ^{\Fibre}(\bxi)+\Q$. Theorems
\ref{Theorem_on_simple_differentiability}, \ref{Theorem_d_jet_homomorphism_extension} and relations (\ref{expression_of_hamiltonian_vector_for_horizontal_observable}) yield that its hamiltonian flow is $d$-jet diffeomorphism. Hence we see, that the hamiltonian flows are homomorphisms in the category of $d$-jet bundles.

\section{Einstein-Cartan Theory}\label{section_ECT}
\subsection{Configuration Bundle of Einstein-Cartan Theory}
In the theory of General Relativity or Einstein-Cartan theory the gravitational interaction is described through evolution of the geometry of the spacetime $\M\simeq\BSigma\times\R$. Therefore we can consider the spacetime only as a topological manifold. For simplicity we assume that $\BSigma$ is compact oriented boundaryless three dimensional manifold.\\
\hs Let $\BSigma_{0}$ be given embedding of $\BSigma$ in $\M$. $\BSigma_{0}$ splits $\M$ into two parts. We choose one and denote its closure as $\M^+_{\BSigma_{0}}$. Closure of the complementar set is denoted as $\M^-_{\BSigma_{0}}$. $\bar{\M}^{\pm}_{\BSigma_{0}}$ denote one point compactifications of $\M^{\pm}_{\BSigma_0}$ and if $\pm\binfty$ are added points called future and past of $\M$ then two point compactification of $\M$ is defined as $\bar{\M}=\{\pm\binfty\}\cup\M$. Diffeomorphism $\sss_{\M}:\BSigma\times\R\to\M$ defines for each $t\in\R$ an embedding $\bt:\Sigma\to\M$ by $\bt(\xxx_{\BSigma})=\sss_{\M}(\xxx_{\BSigma},t)$, where $\xxx_{\BSigma}\in\BSigma$. If $\lim\limits_{t\to\pm\infty}\sss(\xxx_{\BSigma},t)=\pm\binfty$ then the map $\sss_{\M}$ is called slicing of $\M$. We say that curve $\bgm$ on $\M$ is topologically causal, or t-causal, if there exists a slicing $\sss_{\M}$ of $\M$ such that the curve $\bgm$ intersects every $\BSigma_{\bt}=\bt(\BSigma)$ just once. t-causal curve $\bgm$ is called future oriented if $\lim\limits_{s\to\pm\infty}\bgm(s)=\pm\binfty$, where $s$ is a parameter on $\bgm$.\\
\hs Let $\E$ denote a graded bundle of right-handed coframes over $\M$, i.e. the fibre bundle $\proj{\M}{\E}:\E\to\M$ of right-handed basis on the cotangent bundle $\Tan^{\bstar}\M$ defined by
\be
\E=\Union\limits_{\xxx,\e}\e_{\xxx},
\ee
where $\e_{\xxx}=(\xxx,\e^a)$ is right-handed base of the cotangent space $\Tan_{\xxx}^{\bstar}\M$. Its typical fibre $\E^f\simeq\E_{\xxx}=\proj{\M}{\E}^{-1}(\xxx)$ is diffeomorphic with $\GL^+(\R,4)=\{g\in\GL^+(\R,4),\det(g)>0\}$, where the positive general linear group $\GL^+(\R,4)$ is considered as manifold, only. In addition, we also assume that the topological shapes of $\BSigma$ are restricted in such a way that the bundle $\E$ is trivial. General case will be discussed in the forthcomming parts of the series.\\
\hs In the Einstein-Cartan theory $\e\in\E$ is interpreted as an orthonormal coframe settled in the point $\xxx=\proj{\M}{\E}(\e)$. Since the orthonormal coframe $\e=(\xxx,\e^a)$ consists of $\M$-horizontal forms $\e^a$ and $\Tan_{\xxx}\M={\proj{\M}{\E}}_{\bstar}\Tan_{\e}\E$ this defines a metric $\g^{(\e)}$ with Minkowski signature on $\Tan_{\xxx}\M$ by
\ben
\g^{(\e)}({\proj{\M}{\E}}_{\bstar}\vv,{\proj{\M}{\E}}_{\bstar}\ww)=\eta_{ab}\e^a({\proj{\M}{\E}}_{\bstar}\vv)\e^b({\proj{\M}{\E}}_{\bstar}\ww),
\label{ECT_metric}
\een
where $\vv,\ww\in\Tan_{\e}\E$. We can also consider oriented area and volume forms $\blomega^{(\e)}_{ab}$ and $\blomega^{(\e)}_{a}$ in $\xxx\in\M$ 
\ben
\begin{tabular}{lcl}
	$\blomega^{(\e)}_{ab}({\proj{\M}{\E}}_{\bstar}\vv,{\proj{\M}{\E}}_{\bstar}\ww)$&$=$&
	$\left(\frac{1}{2}\bleps_{abcd}\e^c\wedge\e^d\right)(\vv,\ww),$\\
	$\blomega^{(\e)}_{a}({\proj{\M}{\E}}_{\bstar}\uu,{\proj{\M}{\E}}_{\bstar}\vv,{\proj{\M}{\E}}_{\bstar}\ww)$&$=$&
	$\left(\frac{1}{3!}\bleps_{abcd}\e^b\wedge\e^c\wedge\e^d\right)(\uu,\vv,\ww),$
\end{tabular}
\label{ECT_Area_Volume_forms}
\een
where $\uu,\vv,\ww\in\Tan_{\e}\E$.\\
\hs Symmetric bilinear form $\eta_{ab}\e^a\otimes\e^b$ on ${\Tan}_{\e}\E$ is invariant under proper Lorentz transformations $\big(\xxx,\e^a\big)\to\big(\xxx,\OO{a}{\bar{a}}(\e)\e^{\bar{a}}\big)$, i.e. $\det\OO{a}{\bar{a}}(\e)=1$ and $\xxx,\OO{0}{0}(\e)>0$. Such transformation can be exteneded into d-jet diffeomorphism if and only if the matrix $\OO{a}{\bar{a}}(\e)$ depends only on $\xxx$, i.e. $\OO{a}{\bar{a}}(\e)=\OO{a}{\bar{a}}(\xxx)$. Group of such local (proper) Lorentz transformations is denoted by $\SO$.\\
\hs In the standard theory of General Relativity the parallel transport is given by\\ Riemann-Levi-Civita connection, on the other hand the Einstein-Cartan theory works with general metric-compatible connection which should be considered as independent variable. If $\e\in\E$ is orthonormal coframe over the point $\xxx=\proj{\M}{\E}(\e)\in\M$ then the covariant derivative operator $\nabla$ of the metric-compatible connection is given by $\nabla_{\vv}\e^a=-i_{\vv}\AAAA{a}{b}\e^b$ on small neighborhood of $\xxx$, where $\vv\in\Tan_{\xxx}\M$ and $\AAA^{ab}=-\AAA^{ba}=\AAAA{a}{c}\eta^{bc}$ is the connection 1-form.\\
\hs Let $\proj{\E}{\Y}:\Y\to\E$ denote a fibre bundle with coordinates $(\xxx,\e^a,\AAA^{ab})$. Since the bundle $\E\to\M$ is trivial under our assumptions the bundle $\proj{\M}{\E}\circ\proj{\E}{\Y}=\proj{\M}{\Y}:\Y\to\M$ is also trivial. Bundle $\proj{\M}{\Y}:\Y\to\M$ is obviously graded manifold over $\M$.\\
\hs Relation $\nabla_{\vv}\e^a=-i_{\vv}\AAAA{a}{b}\e^b$ yields that the action of the local Lorentz transformation $\OO{}{}\in\SO$ on $\Y$ is given by
\ben
\OO{}{}:(x^{\mu},\e^a,\AAA^{ab})\mapsto(x^{\mu},\OO{a}{\bar{a}}\e^{\bar{a}},
\OO{a}{\bar{a}}\OO{b}{\bar{b}}\AAA^{\bar{a}\bar{b}}+\OO{a}{\bar{a}}\bleta^{\bar{a}\bar{b}}\dd\OO{b}{\bar{b}})\label{Lorentz_action_on_F}.
\een
Let us consider a space $\forms(\Y,\Tan^{\otimes}\M)$ of forms on $\Y$ with values in the total tensor space $\Tan^{\otimes}\M$ over $\M$. We can define an exterior $\SO$-covariant derivative $\DD$ on $\forms(\Y,\Tan^{\otimes}\M)$ by
\be
\DD\uu^a=\dd\uu^a+\AAAA{a}{b}\wedge\uu^b,
\ee
where $\uu^a\in\forms(\Y,\Tan\M)$. Relation $\DD\DD\uu^a=\DD\AAAA{a}{b}\wedge\uu^b$ yields the expression
\be
\DD\AAAA{a}{b}=\dd\AAAA{a}{b}+\AAAA{a}{c}\wedge\AAAA{c}{b}
\ee
of the curvature and Bianchi identity $\DD\DD\AAAA{a}{b}=0$ is also satisfied. It is clear that $\DD\AAA^{ab}$ transforms as a tensor under the action (\ref{Lorentz_action_on_F}) of the local Lorentz group $\SO$.\\
\hs The graded bundle $\proj{\M}{\Y}:\Y\to\M$ is called full configuration bundle of the Einstein-Cartan theory. Its d-jet dual $\JY^{\bstar}$ is equipped with the first canonical form
\ben
\kmef^{(\Y)}=-\p_a\wedge\dd\e^a-\frac{1}{2}\p_{ab}\wedge\dd\AAA^{ab}+\h\label{Full_ECT_kmef}
\een
and the multisymplectic form
\ben
\kmsf^{(\Y)}=-\dd\kmef^{(\Y)}=\dd\p_a\wedge\dd\e^a+\frac{1}{2}\dd\p_{ab}\wedge\dd\AAA^{ab}-\dd\h,\label{Full_ECT_kmsf}
\een
where $\p_a$, $\p_{ab}$ are canonical momenta related to $\e^a$ and $\AAA^{ab}$ and $\h$ is affine hamiltonian coordinate.


\subsection{Equations of Motion, Covariant Legendre Transformation and Reduced Phase Space Bundles}
Let $\M_{\Ii}$ be compact portion of the spacetime bounded by two disjoint embeddings $\BSigma_{\text{ini}}$ and $\BSigma_{\text{fin}}$ and let $\sect\in\Sec(\M,\Y)$ be the section of the configuration bundle $\Y$ such that $\BSigma_{\text{ini}}$ and $\BSigma_{\text{fin}}$ are Cauchy surfaces 
and $\sect^{\bstar}\e^a\vb{\partial\M_{\Ii}}$ are future oriented coframes, i.e. for any t-causal curve $\bgm$ transverse to $\partial\M_{\Ii}$ there holds $\e^0(\dot{\bgm})\vb{\partial\M_{\Ii}}>0$, then the Einstein-Hilbert-Palatini action of the Einstein-Cartan theory is given by 
\ben
\SSS=\int\limits_{\M_{\Ii}}(\djet\sect)^*\LL=
\int\limits_{\M_{\Ii}}\frac{1}{32\pi\varkappa}\bleps_{abcd}\DD_{\sect}\AAA_{\sect}^{ab}\wedge\e_{\sect}^c\wedge\e_{\sect}^d
\label{Lagrangian_ECT_e_A},
\een
where $\varkappa$ is a Newton's constant (the speed of light $c$ is set to $1$), $\e^a_{\sect}=\sect^{\bstar}\e^a$, $\AAA^{ab}_{\sect}=\sect^{\bstar}\AAA^{ab}$ and $\DD_{\sect}$ means $\SO$-covariant exterior derivative on $\M$ given by the connection $\AAA^{ab}_{\sect}$. Variation of the action (\ref{Lagrangian_ECT_e_A})
\be
0=\int\limits_{\M_I}\deltaup\e^a\wedge\left(\frac{1}{16\pi\varkappa}\bleps_{abcd}\DD_{\sect}\AAA_{\sect}^{bc}\wedge\e_{\sect}^d\right)+
\int\limits_{\M_I}\frac{1}{2}\deltaup\AAA^{ab}\wedge\left(-\frac{1}{8\pi\varkappa}\bleps_{abcd}\e_{\sect}^c\wedge\DD_{\sect}\e_{\sect}^d\right),
\ee
where we used boundary conditions $\partial\M_{\Ii}^{\bstar}\big(-\frac{1}{8\pi\varkappa}\bleps_{abcd}\deltaup\AAA^{ab}\wedge\e^c\wedge\e^d)=0$ ($\partial\M_{\Ii}^{\bstar}$ means pullback along embedding of $\partial\M_{\Ii}$ into $\M$), yields equations of motion
\ben
0&=&\frac{1}{16\pi\varkappa}\bleps_{abcd}\DD_{\sect}\AAA_{\sect}^{bc}\wedge\e_{\sect}^d,\label{EOM_ECT_curvature}\\
0&=&-\frac{1}{8\pi\varkappa}\bleps_{abcd}\e_{\sect}^c\wedge\DD_{\sect}\e_{\sect}^d.\label{EOM_ECT_torsion}
\een
The second equation implies that the torsion $\DD_{\sect}\e_{\sect}^a$ of the connection $\AAA_{\sect}^{ab}$ is vanishing on $(\M,\g_{\sect}=\eta_{ab}\e_{\sect}^a\otimes\e_{\sect}^b)$ therefore the first one is equivalent to the Einstein equations of the General Relativity.\\
\hs In section \ref{d_jet} we made assumption that considered Lagrangians are regular, but the Einstein-Hilbert-Palatini Lagrangian $\LL$ given by (\ref{Lagrangian_ECT_e_A}) depends at most linearly on the generalized velocity $\dd\AAA^{ab}$, hence the Lagrangian $\LL$ is singular and we can not express any velocity through the canonical momenta. Instead of exploring the general case of singular Lagrangians we use the following construction. The Legendre transformation (\ref{Legendre_transformation}) $\LTr:\JY\to\JY^{\bstar}$ yields
\ben
\p_a\,\,&=&\phantom{-}0,\label{Legendre_transformation_ECT_I}\\
\p_{ab}&=&\phantom{-}\frac{1}{16\pi\varkappa}\bleps_{abcd}\e^c\wedge\e^d\label{Legendre_transformation_ECT_II},\\
\HH\,\,\,&=&-\frac{1}{32\pi\varkappa}\bleps_{abcd}\AAAA{a}{\bar{a}}\wedge\AAA^{\bar{a}b}\wedge\e^c\wedge\e^d\label{Legendre_transformation_ECT_III}.
\een
The dynamical Cartan-Poincar\' e form on $\Phase=\LTr(\JY)\simeq\Y$ is given by the first equality in (\ref{mef_dynamical})
\ben
\mef=\kmef^{(\Y)}\vb{\Phase}=-\frac{1}{32\pi\varkappa}\bleps_{abcd}\DD\AAA^{ab}\wedge\e^c\wedge\e^d\label{dynamical_Cartan_Poincare_ECT_form},
\een
or if we use embedding $\Hmap:\Phase\to\JE^{\bstar}$
\ben
\e^a&=&\e^a,\label{non_d_jet_transformation}\\
\GGG_a&=&\frac{1}{16\pi\varkappa}\bleps_{abcd}\AAA^{bc}\wedge\e^d\label{canonical_momenta_ECT},\\
\h&=&\HH,
\een
then
\ben
\mef=-\frac{1}{2}\e^a\wedge\dd\GGG_a-\frac{1}{2}\GGG_a\wedge\dd\e^a+\HH(\e,\GGG),\label{dynamical_Cartan_Poincare_ECT_form_G_version}
\een
where the Hamiltonian $\HH(\e,\GGG)$ is given by the expression, see convention of the Hodge operator $*$ in \cite{Fecko} $5.8.1$,
\ben
\HH\equiv\HH(\e,\GGG)=4\pi\varkappa\left[(\GGG_a\wedge\e^b)\wedge*(\GGG_b\wedge\e^a)-\frac{1}{2}(\GGG_a\wedge\e^a)\wedge*(\GGG_b\wedge\e^b)\right].
\label{Hamiltonian_ECT}
\een
The dynamical multisymplectic form is given by 
\ben
\msf=-\dd\mef=-\frac{1}{16\pi\varkappa}\bleps_{abcd}\DD\AAA^{ab}\wedge\e^c\wedge\DD\e^d=
\dd\GGG_a\wedge\dd\e^a-\dd\HH.
\een
This construction shows that the multisymplectic reduction of $\JY^{*}$ defined by (\ref{Legendre_transformation_ECT_I}) and (\ref{Legendre_transformation_ECT_II}) yields that we can use the d-jet dual $\JE^{\bstar}$ of the graded manifold of right-handed coframes $\E$ as the reduced kinematical phase space of the Einstein-Cartan theory if the Cartan-Poincar\' e form on $\JE^{\bstar}$ is given by
\ben
\kmef=-\frac{1}{2}\e^a\wedge\dd\GGG_a-\frac{1}{2}\GGG_a\wedge\dd\e^a+\h=-\GGG_a\wedge\dd\e^a+\h+\dd(\frac{1}{2}\GGG_a\wedge\e^a)
\label{mef_ECT_kinematical}
\een
instead of the first canonical form defined for $\JE^{\bstar}$ by (\ref{mef_dual}). These two forms differ by exact term $\dd(\frac{1}{2}\GGG_a\wedge\e^a)$ therefore the kinematical multisymplectic form of the Einstein-Cartan theory given by
\ben
\kmsf=-\dd\kmef=\dd\GGG_a\wedge\dd\e^a-\dd\h.\label{msf_ECT_kinematical}
\een
and the second canonical form (\ref{msf_dual}) coincide.\\
\hs Lemma \ref{Lemma_on_general_smeared_canonical_observable} yields that following form fields are observables with hamiltonian vectors on $\JE^{\bstar}$
\be
\begin{tabular}{ccc}
	$\e(\bgm)=\bgm_a\wedge\e^a$&$\leftrightarrow$&$\vv(\bgm)=\phantom{-}\bgm_a\partial_{\GGG_a}-(\dd\bgm_a\wedge\e^a)\partial_{\h}$,\\
	$\GGG(\EE)=\GGG_a\wedge\EE^a$&$\leftrightarrow$&$\vv(\EE)=-\EE^a\partial_{\e^a}-(\GGG_a\wedge\dd\EE^a)\partial_{\h}$,
\end{tabular}
\ee
where $\bgm_a$ and $\EE^a$ are smearing $\M$-horizontal $2$- and $1$-forms, respectively, depending only on the spacetime coordinates $x^{\mu}$ and
their non-vanishing local covariant Poisson brackets are
\be
\lpb\e(\bgm),\GGG(\EE)\rpb=\bgm_a\wedge\EE^a.
\ee

There also exist smeared observables related to 2-area and 3-volume. Indeed, it is easy to show that observables $\blomega^{(3)}(\lambda)$ and $\blomega^{(2)}(\BBB)$ given by 
\be
\begin{tabular}{lcl} 
	$\blomega^{(3)}(\lambda)$&$=$&$\frac{1}{3!}\bleps_{abcd}\lambda^a\wedge\e^b\wedge\e^c\wedge\e^d$,\\
	$\blomega^{(2)}(\BBB)$&$=$&$\frac{1}{4}\bleps_{abcd}\BBB^{ab}\wedge\e^c\wedge\e^d$
\end{tabular}
\ee
are related with hamiltonian vectors $\vv(\lambda)$ and $\vv(\BBB)$, where
\be
\begin{tabular}{lcl} 
	$\vv(\lambda)$&$=$&$\Big(\frac{1}{2}\bleps_{abcd}\lambda^a\e^c\wedge\e^d\Big)\partial_{\GGG_b}
	-\left(\frac{1}{3!}\bleps_{abcd}\dd\lambda^a\wedge\e^b\wedge\e^c\wedge\e^d\right)\partial_{\h}$,\\
	$\vv(\BBB)$&$=$&$\Big(\frac{1}{2}\bleps_{abcd}\BBB^{ab}\wedge\e^c\Big)\partial_{\GGG_d}-
	\Big(\frac{1}{4}\bleps_{abcd}\dd\BBB^{ab}\wedge\e^c\wedge\e^d\Big)\partial_{\h}$
\end{tabular}
\ee
and $\lambda^a$, $\BBB^{ab}=-\BBB^{ba}$ are horizontal $0$-, $1$-forms, respectively, depending only on spacetime coordinates $x^{\mu}$.
\subsection{Gauge group of the Einstein-Cartan Theory}
Now, we want to explore the symmetry group of the Einstein-Cartan theory. We employ the coordinates $\e^a, \AAA^{ab}$ rather than the canonical 
pair $\e^a,\GGG_a$. The first observation is that the Lagrangian $\LL$ of the Einstein-Cartan theory is invariant under the action (\ref{Lorentz_action_on_F}) of the local Lorentz group $\SO$. Its infinitesimal action is given by
\ben\label{Lorentz_action_infinitesimal_version}
\begin{array}{ccl}
	\bar{\e}^a&=&\e^a-\Lambda^{a}_{\phantom{a}b}\e^b\\
	\bar{\AAA}^{ab}&=&\AAA^{ab}-\Lambda^{a}_{\phantom{a}\bar{a}}\AAA^{\bar{a}b}-\Lambda^{b}_{\phantom{b}\bar{b}}\AAA^{a\bar{b}}+\dd\Lambda^{ab}
	=\AAA^{ab}+\DD\Lambda^{ab},
\end{array}
\een
where $\Lambda^{ab}=-\Lambda^{ba}$ are $0$-forms on $\M$ playing role of the coordinates on the algebra of the local Lorentz group $\SO$. 
The generating vector of the infinitesimal action (\ref{Lorentz_action_infinitesimal_version}) is given by
\ben
\vv(\Lambda)=-\Lambda^{a}_{\phantom{a}b}\e^b\partial_{\e^a}+\frac{1}{2}\DD\Lambda^{ab}\partial_{\AAA^{ab}}
\een
Its basic properties are summarized in the following lemma.
\begin{lemma}$\phantom{.}$\\
	\hs $\phantom{(ii}i)$ \hs \begin{minipage}[t]{0.75\textwidth}The dynamical Cartan-Poincar\' e form (\ref{dynamical_Cartan_Poincare_ECT_form}) is invariant under the action (\ref{Lorentz_action_on_F}) of the local Lorentz group $\SO$ therefore
		\be\lder{\vv(\Lambda)}\mef=0.\ee
	\end{minipage}\\\\
	\hs $\phantom{(i}ii)$ \hs\begin{minipage}[t]{0.75\textwidth} Noether's current $\ttt(\Lambda):=i_{\vv(\Lambda)}\mef$ is given by formula
		\be 
		\ttt(\Lambda):=-\frac{1}{32\pi\varkappa}\bleps_{abcd}\DD\Lambda^{ab}\wedge\e^c\wedge\e^d.\ee
	\end{minipage}\\\\
	\hs $\phantom{(}iii)$ \hs \begin{minipage}[t]{0.75\textwidth} Vanishing of the Noether's charge
		\[
		\TTT^{\sect}_{\BSigma_{E}}(\Lambda)=\int\limits_{\BSigma_{E}}\sect^*\ttt(\Lambda)=0
		\] 
		for all embeddings $\BSigma_{\text{E}}\simeq\BSigma$ and all $\Lambda^{ab}\in\alg{\SO}$ implies that $\sect\in\Sec(\M,\Phase)$ is a 
		solution of the Hamilton equation (\ref{EOM_ECT_torsion}).\end{minipage}
\end{lemma}
\begin{proof}
	Let $\OO{}{}(s\Lambda)\in\SO$ denote one parametric group associated to the vector $\vv(\Lambda)$. Direct calculation of its action 
	on $\mef$ yields
	\be
	\bar\mef&=&-\frac{1}{32\pi\varkappa}\bleps_{abcd}\bar{\DD}\bar{\AAA}^{ab}\wedge\bar{\e}^c\wedge\bar{\e}^d\\
	&=&-\frac{1}{32\pi\varkappa}\bleps_{abcd}
	\OO{a}{\bar{a}}(s\Lambda)\OO{b}{\bar{b}}(s\Lambda)\OO{c}{\bar{c}}
	(s\Lambda)\OO{d}{\bar{d}}(s\Lambda)\DD\AAA^{\bar{a}\bar{b}}\wedge\e^{\bar{c}}\wedge\e^{\bar{d}}\\
	&=&-\det\left(\OO{\bar{a}}{\bar{b}}(s\Lambda)\right)\frac{1}{32\pi\varkappa}\bleps_{abcd}\DD\AAA^{ab}\wedge\e^c\wedge\e^d=\mef,
	\ee
	since $\det\left(\OO{\bar{a}}{\bar{b}}(s\Lambda)\right)=1$. Therefore we have 
	\be
	\lder{\vv(\Lambda)}\mef=\lim_{s\to 0}\frac{\bar{\mef}-\mef}{s}=0.
	\ee
	The interior product of $\vv(\Lambda)$ with the Cartan-Poincar\' e form $\mef$ is
	\be
	\ttt(\Lambda)=i_{\vv(\Lambda)}\mef=-\frac{1}{32\pi\varkappa}\bleps_{abcd}\DD\Lambda^{ab}\wedge\e^c\wedge\e^d
	\ee
	and the integration over given emmbeding $\BSigma_{\text{E}}$ yields relation
	\be
	\TTT^{\sect}_{\BSigma_{E}}(\Lambda)&=&\int\limits_{\BSigma_{E}}\sect^*\ttt(\Lambda)=
	\int\limits_{\BSigma_{E}}-\frac{1}{32\pi\varkappa}\bleps_{abcd}\DD_{\sect}\Lambda^{ab}\wedge\e_{\sect}^c\wedge\e_{\sect}^d\\
	&=&\int\limits_{\BSigma_{E}}\frac{1}{2}\Lambda^{ab}
	\left(\frac{1}{8\pi\varkappa}\bleps_{abcd}\e_{\sect}^c\wedge\DD_{\sect}\e_{\sect}^d\right)+
	\oint\limits_{\partial\BSigma_{E}}-\frac{1}{32\pi\varkappa}\bleps_{abcd}\Lambda^{ab}\e_{\sect}^c\wedge\e_{\sect}^d.
	\ee
	The last term is vanishing since $\partial\BSigma_{E}=\emptyset$. The condition $\TTT^{\sect}_{\BSigma_{\text{E}}}(\Lambda)=0$ for all 
	$\Lambda^{ab}$ implies that $-\frac{1}{8\pi\varkappa}\bleps_{abcd}\e_{\sect}^c\wedge\DD_{\sect}\e_{\sect}^d=0$ on the embedding 
	$\BSigma_{\text{E}}$ and arbitrariness of the embedding yields that the equation (\ref{EOM_ECT_torsion}) is satisfied on whole $\M$ and 
	vice versa if the equation (\ref{EOM_ECT_torsion}) is satisfied then $\TTT^{\sect}_{\BSigma_{\text{E}}}(\Lambda)=0$ for all embeddings $\BSigma_{\text{E}}$ and all $\Lambda^{ab}$.
\qed\end{proof}
Einstein's principle \cite{Einstein_MoR} of General Relativity says that the laws of the Nature are independent on the choice of the spacetime coordinates we are using for the description of the physical system. Therefore the action (\ref{Lagrangian_ECT_e_A}) should be invariant under the action of the group of diffeomorphisms $\DiffM$. Let $\bleta_{\M}(s)\in\DiffM$ denote one parameter group generated by the spacetime vector $\bxi\in\Tan\M$. Its prolongation $\bleta_{\Phase}(s)$ to the fibre bundle $\Phase\to\M$ is given by 
\ben
\bleta_{\Phase}:(x^{\mu},\e^a,\AAA^{ab})\mapsto\left(\eta_{\M}^{\mu},(\bleta_{\M}^{-1})^*\e^a,(\bleta_{\M}^{-1})^*\AAA^{ab}\right)
\label{diff_action_on_F}.
\een
The generating vector $\ww(\bxi)$ of $\bleta_{\Phase}$ on $\Phase$ is given by an expression
\ben
\ww(\bxi)=\bxi-i_{\dd\bxi}\e^a\partial_{\e^a}-\frac{1}{2}i_{\dd\bxi}\AAA^{ab}\partial_{\AAA^{ab}}\label{diff_generator}.
\een
We have similar lemma for $\ww(\bxi)$ as for $\vv(\Lambda)$.
\begin{lemma}
	\be
	\begin{tabular}{r p{9.5cm}}
		$\phantom{(}i)$&The dynamical Cartan-Poincar\' e form (\ref{dynamical_Cartan_Poincare_ECT_form}) is invariant under the 
		action (\ref{diff_action_on_F}) of the diffeomorphism group $\DiffM$ therefore
		\[\lder{\ww(\bxi)}\mef=0.\]\\
		$\phantom{(}ii)$&Noether's current $\rr(\bxi):=i_{\ww(\bxi)}\mef$ is given by formula
		\be 
		\rr(\bxi):&=&
		\phantom{+} \frac{1}{16\pi\varkappa}\bleps_{abcd}\DD\AAA^{ab}\wedge\e^c i_{\bxi}\e^d+
		\frac{1}{16\pi\varkappa}\bleps_{abcd}i_{\bxi}\AAA^{ab}\wedge\e^c\wedge\DD\e^d\\
		&&+\,\dd\left(\frac{1}{32\pi\varkappa}\bleps_{abcd}i_{\bxi}\AAA^{ab}\wedge\e^c\wedge\e^d\right)	
		\ee
		\\
		$\phantom{(}iii)$&
		If for the section $\sect\in\Sec(\M,\Phase)$ the Hamilton equation (\ref{EOM_ECT_torsion}) is satisfied then the vanishing of 
		the Noether's charge
		\[
		\RR_{\BSigma_{E}}^{\sect}(\bxi)=\int\limits_{\BSigma_{E}}\sect^*\rr(\bxi)=0
		\] 
		for all embeddings $\BSigma_{\text{E}}\simeq\BSigma$ and all $\bxi\in\alg{\DiffM}$ implies that $\sect\in\Sec(\M,\Phase)$ is also a solution 
		of the Hamilton equation (\ref{EOM_ECT_curvature}).
	\end{tabular}
	\ee
\end{lemma}
\begin{proof}
	At first we need to calculate some auxiliar relations. We employ notation $\y^A=(\e^a,\AAA^{ab})$. Cartan-Poincar\' e form $\mef$ can 
	be written as
	\be
	\mef=\dd y^A_{\mu}\wedge\Theta^{\mu\nu}_{A}\dd\Sigma_{\nu}+\tilde{H}\dd\Sigma,
	\ee
	where $\Theta^{\mu\nu}_{A}$ and $\tilde{H}$ are certain functions on $\Phase$. Since $\mef$ does not depend explicitly on the 
	coordinate $x^{\mu}$ we have for $\bxi=\xi^{\mu}\partial_{\mu}$ followig relation
	\be
	i_{\bxi}\dd\mef+\dd i_{\bxi}\mef&=&
	-\dd y^A_{\mu}\wedge\dd\Theta^{\mu\nu}\wedge i_{\bxi}\dd\Sigma_{\nu}-\dd\tilde{H}\wedge i_{\bxi}\dd\Sigma+\\
	&&+\dd y^A_{\mu}\wedge\dd\Theta^{\mu\nu}\wedge i_{\bxi}\dd\Sigma_{\nu}+
	\dd y^A_{\mu}\wedge\Theta^{\mu\nu}\wedge\dd(i_{\bxi}\Sigma_{\nu})+\\
	&&+\dd\tilde{H}\wedge i_{\bxi}\dd\Sigma+\tilde{H}\dd(i_{\bxi}\dd\Sigma)
	\ee
	or after simplification
	\ben
	i_{\bxi}\dd\mef+\dd i_{\bxi}\mef=i_{\dd\bxi}\mef\label{auxiliary_relation_i}.
	\een
	Similar calculation yields
	\ben
	i_{\bxi}\dd\y^A+\dd i_{\bxi}\y^A=i_{\dd\bxi}\y^A\label{auxiliary_relation_ii}.
	\een
	Now, we can prove the first indentity$\phantom{(}i)$ of the lemma
	\be
	\lder{\ww(\bxi)}\mef&=&i_{\ww(\bxi)}\dd\mef+\dd i_{\ww(\bxi)}\mef\\
	&=&\phantom{+}i_{\bxi}\dd\mef-\frac{1}{16\pi\varkappa}\bleps_{abcd}i_{\dd\bxi}\AAA^{ab}\wedge\e^c\wedge\DD\e^d
	+\frac{1}{16\pi\varkappa}\bleps_{abcd}\DD\AAA^{ab}\wedge\e^c\wedge i_{\dd\bxi}\e^d+\\
	&&+\dd i_{\bxi}\mef+\frac{1}{32\pi\varkappa}\dd\left(\bleps_{abcd}i_{\dd\bxi}\AAA^{ab}\wedge\e^c\wedge\e^d\right).
	\ee 
	The first auxiliar formula (\ref{auxiliary_relation_i}) yields
	\be
	\lder{\ww(\bxi)}\mef&=&\frac{1}{16\pi\varkappa}\bleps_{abcd} i_{\dd\bxi}\AAA^{a\bar{a}}\wedge\AAAA{b}{\bar{a}}\wedge\e^c\wedge\e^d-
	\frac{1}{16\pi\varkappa}\bleps_{abcd}i_{\dd\bxi}\AAA^{ab}\wedge\e^c\wedge\AAAA{d}{\bar{d}}\wedge\e^{\bar{d}}.
	\ee
	The first term can be rewritten as
	\ben
	\frac{1}{16\pi\varkappa}\bleps_{abcd} i_{\dd\bxi}\AAA^{a\bar{a}}\wedge\AAAA{b}{\bar{a}}\wedge\e^c\wedge\e^d=
	\frac{1}{16\pi\varkappa}\bleps_{abcd}
	\frac{1}{4}\bar{\bleps}^{a\bar{a}\hat{c}\hat{d}}\bleps_{\hat{a}\hat{b}\hat{c}\hat{d}}
	i_{\dd\bxi}\AAA^{\hat{a}\hat{b}}\wedge\AAAA{b}{\bar{a}}\wedge\e^c\wedge\e^d\nonumber\\
	=\frac{1}{32\pi\varkappa}\bleps_{\hat{a}\hat{b}\hat{c}\hat{d}}
	(\delta^{\bar{a}}_b\delta_{cd}^{\hat{c}\hat{d}}+\delta_c^{\bar{a}}\delta_{db}^{\hat{c}\hat{d}}+\delta_{d}^{\bar{a}}\delta_{bc}^{\hat{c}\hat{d}})
	i_{\dd\bxi}\AAA^{\hat{a}\hat{b}}\wedge\AAAA{b}{\bar{a}}\wedge\e^c\wedge\e^d\;\;\;\;\;\;\;\;
	\label{example_of_calculation}\\
	=\frac{1}{16\pi\varkappa}\bleps_{abcd}i_{\dd\bxi}\AAA^{ab}\wedge\e^c\wedge\AAAA{d}{\bar{d}}\wedge\e^{\bar{d}}
	\;\;\;\;\;\;\;\;\;\;\;\;\;\;\;\;\,\;\;\;\;\;\;\;\;\;\;\;\;\;\;\;\;\;\;\;\;\;\;\;\;\nonumber
	\een
	which yields
	\be
	\lder{\ww(\bxi)}\mef=0.
	\ee
	Noether's current is given by
	\be
	\rr(\bxi)&=\phantom{+}&i_{\ww(\bxi)}\mef=i_{\bxi}\mef+\frac{1}{32\pi\varkappa}\bleps_{abcd}i_{\dd\bxi}\AAA^{ab}\wedge\e^c\wedge\e^d\\
	&=&-\frac{1}{32\pi\varkappa}\bleps_{abcd}i_{\bxi}\DD\AAA^{ab}\wedge\e^c\wedge\e^d+
	\frac{1}{16\pi\varkappa}\bleps_{abcd}\DD\AAA^{ab}\wedge\e^ci_{\bxi}\e^d\\
	&&+\frac{1}{32\pi\varkappa}
	\bleps_{abcd}i_{\dd\bxi}\AAA^{ab}\wedge\e^c\wedge\e^d.
	\ee
	Now, if we use the second auxiliar relation (\ref{auxiliary_relation_ii}) then the similar calculation as in (\ref{example_of_calculation}) yields 
	\be
	\rr(\bxi)&=&\phantom{+}\frac{1}{16\pi\varkappa}\bleps_{abcd}\DD\AAA^{ab}\wedge\e^c i_{\bxi}\e^d+
	\frac{1}{16\pi\varkappa}\bleps_{abcd}i_{\bxi}\AAA^{ab}\wedge\e^c\wedge\DD\e^d\\
	&&+\dd\left(\frac{1}{32\pi\varkappa}\bleps_{abcd}i_{\bxi}\AAA^{ab}\wedge\e^c\wedge\e^d\right).
	\ee
	\\
	Noether's charge on embedding $\BSigma_{\text{E}}$ for the section $\sect\in\Sec(\M,\Phase)$ is\\
	\be
	\RR^{\sect}_{\BSigma_{\text{E}}}(\bxi)&=&\int\limits_{\BSigma_{\text{E}}}\sect^*\rr(\bxi)\\
	&=&\int\limits_{\BSigma_{\text{E}}}
	\frac{1}{16\pi\varkappa}\bleps_{abcd}\DD_{\sect}\AAA_{\sect}^{ab}\wedge\e_{\sect}^c i_{\bxi}\e_{\sect}^d+
	\int\limits_{\BSigma_{\text{E}}}
	\frac{1}{16\pi\varkappa}\bleps_{abcd}i_{\bxi}\AAA_{\sect}^{ab}\wedge\e_{\sect}^c\wedge\DD_{\sect}\e_{\sect}^d.
	\ee
	If $\sect$ is the solution of the Hamilton equation (\ref{EOM_ECT_torsion}) then the second term is vanishing and the arbitrariness of the 
	embedding $\BSigma_{\text{E}}$ and $\bxi$ implies that the vanishing of the charge is equivalent to the Hamilton equation (\ref{EOM_ECT_curvature}).
\qed\end{proof}
$\phantom{,.}$\\
As a direct consequence of these two lemmata we have that the mutual vanishing of the Noether's charges $\TTT(\Lambda)$ and $\RR(\bxi)$ is 
equivalent to the Hamilton equations (\ref{EOM_ECT_curvature}) and (\ref{EOM_ECT_torsion}). Therefore it seems that the symmetry group $\GG$ of the Einstein-Cartan theory can be constructed from groups $\SO$ and $\DiffM$. Indeed, we finally find out that the symmetry group is given by the semidirect product $\GG=\SO\rtimes\DiffM$.\\
\hs Before we do that let us at first show, that the linear span of $\alg{\SO}$ and $\alg{\DiffM}$ is closed subalgebra in Lie algebra of vector fields $(\Tan\Phase,\lsn,\rsn)$.\\
\\
\begin{lemma}
	$\Span\left(\alg{\SO}, \alg{\DiffM}\right)$ is a Lie algebra and the commutators of the generators are
	
	\begin{tabular}{r p{9.5cm}}
		$\phantom{(}i)$&\[\lsn\vv(\Lambda),\vv(\Lambda ')\rsn=\vv(\Lambda\bleta\Lambda '-\Lambda '\bleta\Lambda),\]\\
		$\phantom{(}ii)$&\[\lsn\ww(\bxi),\ww(\bxi ')\rsn=\ww(\lder{\bxi}\bxi '),\;\;\;\;\;\;\;\;\;\;\;\;\;\]\\
		$\phantom{(}iii)$&\[\lsn\ww(\bxi),\vv(\Lambda)\rsn=\vv(\lder{\bxi}\Lambda).\;\;\;\;\;\;\;\;\;\;\;\;\;\]
	\end{tabular}
\end{lemma}
\begin{proof}
	We have for$\phantom{(}i)$ an expression
	\be
	\lsn\vv(\Lambda),\vv(\Lambda ')\rsn&=&
	\lsn-\Lambda^{a}_{\phantom{a}b}\e^b\partial_{\e^a},-\Lambda^{'c}_{\phantom{'c}d}\e^d\partial_{\e^c}\rsn+
	\lsn\frac{1}{2}\DD\Lambda^{ab}\partial_{\AAA^{ab}},\frac{1}{2}\DD\Lambda^{'cd}\partial_{\AAA^{cd}}\rsn.
	\ee
	The first bracket can be evaluated as 
	\be
	\lsn-\Lambda^{a}_{\phantom{a}b}\e^b\partial_{\e^a},-\Lambda^{'c}_{\phantom{'c}d}\e^d\partial_{\e^c}\rsn&=&
	\Lambda^{a}_{\phantom{a}b}\e^b\partial_{\e^a}\Lambda^{'c}_{\phantom{'c}d}\e^d\partial_{\e^c}-
	\Lambda^{'c}_{\phantom{'c}d}\e^d\partial_{\e^c}\Lambda^{a}_{\phantom{a}b}\e^b\partial_{\e^a}\\
	&=&-(\Lambda^{a}_{\phantom{a}c}\Lambda^{'c}_{\phantom{'c}b}-\Lambda^{'a}_{\phantom{'a}c}\Lambda^{c}_{\phantom{c}b})\e^b\partial_{\e^a}.
	\ee
	Evaluation of the second term yields
	\begin{align*}
	&\lsn\frac{1}{2}\DD\Lambda^{ab}\partial_{\AAA^{ab}},\frac{1}{2}\DD\Lambda^{'cd}\partial_{\AAA^{cd}}\rsn\;=\\
	&=\phantom{+}
	\lsn\frac{1}{2}\dd\Lambda^{ab}\partial_{\AAA^{ab}},\frac{1}{2}\dd\Lambda^{'cd}\partial_{\AAA^{cd}}\rsn+
	\lsn\frac{1}{2}\dd\Lambda^{ab}\partial_{\AAA^{ab}},-\Lambda^{'c}_{\phantom{'c}\bar{c}}\AAA^{\bar{c}d}\partial_{\AAA^{cd}}\rsn+\\
	&\phantom{=\;}+\lsn-\Lambda^{a}_{\phantom{a}\bar{a}}\AAA^{\bar{a}b}\partial_{\AAA^{ab}},\frac{1}{2}\dd\Lambda^{'cd}\partial_{\AAA^{cd}}\rsn+
	\lsn-\Lambda^{a}_{\phantom{a}\bar{a}}\AAA^{\bar{a}b}\partial_{\AAA^{ab}},-\Lambda^{'c}_{\phantom{'c}\bar{c}}\AAA^{\bar{c}d}
	\partial_{\AAA^{cd}}\rsn\\
	&=
	-\dd\Lambda^{cb}\Lambda^{'a}_{\phantom{'a}c}\partial_{\AAA^{ab}}+\dd\Lambda^{'cb}\Lambda^{a}_{\phantom{a}c}\partial_{\AAA^{ab}}
	-(\Lambda^{a}_{\phantom{a}c}\Lambda^{'c}_{\phantom{'c}\bar{c}}-\Lambda^{'a}_{\phantom{'a}c}\Lambda^{c}_{\phantom{c}\bar{c}})
	\AAA^{\bar{c}b}\partial_{\AAA^{ab}}\\
	&=\frac{1}{2}\dd(\Lambda^{a}_{\phantom{a}c}\Lambda^{'cb}-\Lambda^{'a}_{\phantom{'a}c}\Lambda^{cb})
	\partial_{\AAA^{ab}}-(\Lambda^{a}_{\phantom{a}c}\Lambda^{'c}_{\phantom{'c}\bar{c}}-\Lambda^{'a}_{\phantom{'a}c}\Lambda^{c}_{\phantom{c}\bar{c}})
	\AAA^{\bar{c}b}\partial_{\AAA^{ab}}\\
	&=\frac{1}{2}\DD(\Lambda^{a}_{\phantom{a}c}\Lambda^{'cb}-\Lambda^{'a}_{\phantom{'a}c}\Lambda^{cb})
	\partial_{\AAA^{ab}}.
	\end{align*}
	These two relation establish the first equality$\phantom{(}i)$ of the lemma.\\
	\hs In order to prove the second relation $\phantom{(}ii)$ we employ for a while the notation $\y^A=(\e^a,\AAA^{ab})$ again. 
	The generator $\ww(\bxi)$ can be rewritten as 
	\be
	\ww(\bxi)=\bxi-i_{\dd\bxi}\y^A\partial_{\y^A}.
	\ee
	Thus, we have
	\be
	\lsn\ww(\bxi),\ww(\bxi ')\rsn&=&\lder{\bxi}\bxi '-\xi^{\nu}\xi^{'\lambda}_{,\mu\nu}y^{A}_{\lambda}\partial_{y^{A}_{\mu}}+
	\xi^{'\nu}\xi^{\lambda}_{,\mu\nu}y^{A}_{\lambda}\partial_{y^{A}_{\mu}}\\
	&&\phantom{\lder{\bxi}\bxi '}-\xi^{\nu}_{,\lambda}\xi^{'\lambda}_{,\mu}y^{A}_{\nu}\partial_{y^{A}_{\mu}}
	+\xi^{'\nu}_{,\lambda}\xi^{\lambda}_{,\mu}y^{A}_{\nu}\partial_{y^{A}_{\mu}}\\
	&=&\lder{\bxi}\bxi '-i_{\dd\lder{\bxi}\bxi '}\y^A\partial_{\y^A}\\
	&=&\ww(\lder{\bxi}\bxi '),
	\ee
	which proves$\phantom{(}ii)$.\\
	\hs The last statement$\phantom{(}iii)$ is given by calculation
	\be
	\lsn\ww(\bxi),\vv(\Lambda)\rsn\phantom{+}
	&=&\phantom{+}\lsn\bxi,-\Lambda^c_{\phantom{c}d}\e^d\partial_{\e^c}\rsn+\lsn\bxi,\frac{1}{2}\DD\Lambda^{cd}\partial_{\AAA^{bc}}\rsn+\\
	&&+\lsn-i_{\dd\bxi}\e^a\partial_{\e^a},-\Lambda^c_{\phantom{c}d}\e^d\partial_{\e^c}\rsn+
	\lsn-i_{\dd\bxi}\e^a\partial_{\e^a},\frac{1}{2}\DD\Lambda^{cd}\partial_{\AAA^{bc}}\rsn+\\
	&&+\lsn-\frac{1}{2}i_{\dd\bxi}\AAA^{ab}\partial_{\AAA^{ab}},-\Lambda^c_{\phantom{c}d}\e^d\partial_{\e^c}\rsn+
	\lsn-\frac{1}{2}i_{\dd\bxi}\AAA^{ab}\partial_{\AAA^{ab}},\frac{1}{2}\DD\Lambda^{cd}\partial_{\AAA^{bc}}\rsn.
	\ee
	The fourth and fifth terms are obviously vanishing. The third term is also vanishing because
	\be
	\lsn-i_{\dd\bxi}\e^a\partial_{\e^a},-\Lambda^c_{\phantom{c}d}\e^d\partial_{\e^c}\rsn=\Lambda^c_{\phantom{c}d}i_{\dd\bxi}\e^d\partial_{\e^c}-
	i_{\dd\bxi}(\Lambda^c_{\phantom{c}d}\e^d)\partial_{\e^c}=0.
	\ee
	Calculation of the remaining terms yields expressions:\\
	the first term
	\be
	\lsn\bxi,-\Lambda^c_{\phantom{c}d}\e^d\partial_{\e^c}\rsn=-\lder{\bxi}\Lambda^a_{\phantom{a}b}\e^b\partial_{\e^a},
	\ee
	the second term
	\be
	\lsn\bxi,\frac{1}{2}\DD\Lambda^{cd}\partial_{\AAA^{bc}}\rsn=\frac{1}{2}(\lder{\bxi}-i_{\dd\bxi})\dd\Lambda^{ab}\partial_{\AAA^{ab}}
	-\lder{\bxi}\Lambda^a_{\phantom{a}\bar{a}}\AAA^{\bar{a}b}\partial_{\AAA^{ab}}
	\ee
	and the sixth term
	\be
	\lsn-\frac{1}{2}i_{\dd\bxi}\AAA^{ab}\partial_{\AAA^{ab}},\frac{1}{2}\DD\Lambda^{cd}\partial_{\AAA^{bc}}\rsn=
	\frac{1}{2}i_{\dd\bxi}\dd\Lambda^{ab}\partial_{\AAA^{ab}}.
	\ee
	Therefore
	\be
	\lsn\ww(\bxi),\vv(\Lambda)\rsn=-\lder{\bxi}\Lambda^a_{\phantom{a}b}\e^b\partial_{\e^a}+\frac{1}{2}\DD\lder{\bxi}\Lambda^{ab}\partial_{\AAA^{ab}}
	=\vv(\lder{\bxi}\Lambda).
	\ee
\qed\end{proof}
Now, we want to explore transformation on $\Phase$ given by composition $\OO{}{}(\Lambda)\circ \bleta_{\Phase}(\bxi)$. From definitions (\ref{Lorentz_action_on_F}) and (\ref{diff_action_on_F}) we have
\begin{align*}
&\OO{}{}(\Lambda)\circ\bleta_{\Phase}(\bxi):(x^{\mu},\e^a,\AAA^{ab})\mapsto
(\eta_{\M}^{\mu}, \OO{a}{b}(\eta^{-1}_{\M})^*\e^b,\OO{a}{\bar{a}}\OO{b}{\bar{b}}(\eta^{-1}_{\M})^*\AAA^{\bar{a}\bar{b}}+
\OO{a}{\bar{a}}\eta^{\bar{a}\bar{b}}\dd\OO{b}{\bar{b}})=\\
&=
\left(\eta_{\M}^{\mu}, (\eta^{-1}_{\M})^*\left((\eta_{\M})^*\OO{a}{b}\e^b\right),
(\eta^{-1}_{\M})^*\left((\eta_{\M})^*\OO{a}{\bar{a}}(\eta_{\M})^*\OO{b}{\bar{b}}\AAA^{\bar{a}\bar{b}}+
(\eta_{\M})^*\OO{a}{\bar{a}}\eta^{\bar{a}\bar{b}}\dd(\eta_{\M})^*\OO{b}{\bar{b}}\right)\right)\\
&=\bleta_{\Phase}(\bxi)\circ\OO{}{}\left((\eta_{\M})^*\Lambda\right):(x^{\mu},\e^a,\AAA^{ab}),
\end{align*}
which yields an equality
\be
\OO{}{}(\Lambda)\circ\bleta_{\Phase}(\bxi)=\bleta_{\Phase}(\bxi)\circ\OO{}{}\left((\eta_{\M})^*\Lambda\right).
\ee
This implies that if we define $\GG$ as a set of all elements of the type $\OO{}{}(\Lambda)\circ\bleta_{\Phase}(\bxi)$ then $(\GG,\circ)$ is a 
group given by semidirect product $\GG=\SO\rtimes\DiffM$, where $\SO$ is normal subgroup of $\GG$. Since $\GG$ is localizable the symmetry group is a gauge group.\\
\hs Finally, we can formulate the main theorem of this subsection. 
\begin{theorem}
	Group $\GG$ given by semidirect product $\GG=\SO\rtimes\DiffM$ is a gauge group of the Einstein-Cartan theory. Vanishing of all 
	Noether's charges associated to the generators of $\GG$ on $\Phase$ is equivalent to the Hamilton equations of motion 
	(\ref{EOM_ECT_curvature}) and (\ref{EOM_ECT_torsion}). The local Poisson brackets among Noether's currents $\ttt(\Lambda)$
	and $\rr(\bxi)$ are\\
	\begin{tabular}{r p{9.5cm}}
		$\phantom{(}i)$&\[\Lpb\ttt(\Lambda),\ttt(\Lambda ')\Rpb=
		\ttt(\Lambda\bleta\Lambda '-\Lambda '\bleta\Lambda),\]\\
		$\phantom{(}ii)$&\[\Lpb\rr(\bxi),\rr(\bxi ')\Rpb=
		\rr(\lder{\bxi}\bxi '),\;\;\;\;\;\;\;\;\;\;\;\;\;\]\\
		$\phantom{(}iii)$&\[\Lpb\rr(\bxi),\ttt(\Lambda)\Rpb=
		\ttt(\lder{\bxi}\Lambda).\;\;\;\;\;\;\;\;\;\;\;\;\;\]
	\end{tabular}
\end{theorem}
\begin{proof}
	Statements of the theorem are direct consequencies of the previous two lemmata, the relation (\ref{local_Poisson_bracket_Noethers_charges_general_form}) and discussion in the text above.
\qed\end{proof}


\section{Conclusion}
The main goal of the series is to propose as hypothesis a new theory of covariant quantum gra\-vi\-ty with continuous quantum geometry.
The first part of the series was dealing with the covariant hamiltonian formulation of the Einstein-Cartan theory. We found out that the kinematical
phase space is given as a d-jet dual $\JE^{\bstar}$ of the graded manifold of right-handed coframes over the spacetime manifold $\M$ with the Cartan-Poincar\' e form slightly modified by addition of the exact term, which is an artefact of the multisymplectic reduction since the Lagrangian of the Einstein-Cartan theory is singular. We showed that the gauge group $\GG$ of the Einstein-Cartan theory is given by the semidirect product $\GG=\SO\rtimes\DiffM$ of the local (proper) Lorentz group $\SO$ and the group of spacetime diffeomorphisms $\DiffM$ hence the local Poisson algebra of its generators is closed Lie algebra. Vanishing of all Noether's charges is equivalent to equations of motion of the Einstein-Cartan theory. \\
\hs These suggest that we solved the old Kucha\v r's problem of finding the formulation of constraints which forms closed Lie algebra, but the opposite is truth. We should take into account, that we just proved that only the local algebra is closed. In order to find an integral version of just obtained results we should, at first, introduce an instantaneous formalism what is the goal of the next part of the series. 

\bibliographystyle{spmpsci}
\bibliography{CQG_I}
\end{document}